\documentclass[11pt]{article}
\usepackage[utf8]{inputenc}
\usepackage{amsmath,amsfonts,amsthm,amssymb,geometry,bbm}
\usepackage{graphics,color}
\usepackage{epsfig}
\usepackage{caption}
\usepackage{algorithmic}
\usepackage[ruled]{algorithm2e}
\usepackage{blkarray}
\usepackage[T1]{fontenc}
\usepackage{enumitem}
\usepackage[numbers]{natbib}
\usepackage{mathtools}
\usepackage{cancel}
\usepackage{thm-restate}
\usepackage{rotating}
\usepackage{multirow}

\newtheorem{theorem}{Theorem}[section]

\newtheorem{lemma}[theorem]{Lemma}
\newtheorem{proposition}[theorem]{Proposition}
\newtheorem{conjecture}[theorem]{Conjecture}

\theoremstyle{definition}
\newtheorem{remark}[theorem]{Remark}

\newtheorem{definition}[theorem]{Definition}
\newtheorem{example}[theorem]{Example}

\newcommand{\infmax}{{\sc InfMax}\xspace}
\newcommand{\maxc}{{\sc Max-k-Coverage}\xspace}
\DeclareMathOperator*{\E}{\mathbb{E}}

\DeclareMathOperator*{\argmax}{argmax}

\DeclareMathOperator*{\val}{val}

\newcommand{\Z}{{\mathbb Z}}

\newcommand{\calS}{{\mathcal S}}
\newcommand{\calP}{{\mathcal P}}
\newcommand{\M}{{\mathcal M}}
\newcommand{\IN}{{\textit IN}}

\newcommand{\infect}[1]{\xrightarrow[]{\cancel{#1}}}

\title{Limitations of Greed: Influence Maximization in Undirected Networks Re-visited\footnote{A short version of this paper is appeared in AAMAS'20. Grant Schoenebeck, Biaoshuai Tao, and Fang-Yi Yu are pleased to acknowledge the support of National Science Foundation AitF \#1535912 and CAREER \#1452915.}}
\author{Grant Schoenebeck\thanks{University of Michigan, School of Information, schoeneb@umich.edu} \and Biaoshuai Tao\thanks{University of Michigan, Division of Computer Science and Engineering, bstao@umich.edu} \and Fang-Yi Yu\thanks{University of Michigan, School of Information, fayu@umich.edu}}
\date{}

\begin{document}

\maketitle

\begin{abstract}

We consider the influence maximization problem (selecting $k$ seeds in a network maximizing the expected total influence) on undirected graphs under the linear threshold model.  On the one hand, we prove that the greedy algorithm always achieves a $(1 - (1 - 1/k)^k + \Omega(1/k^3))$-approximation, showing that the greedy algorithm does slightly better on undirected graphs than the generic $(1- (1 - 1/k)^k)$ bound which also applies to directed graphs.  On the other hand, we show that substantial improvement on this bound is impossible by presenting an example where the greedy algorithm can obtain at most a $(1- (1 - 1/k)^k + O(1/k^{0.2}))$ approximation.

This result stands in contrast to the previous work on the independent cascade model.
Like the linear threshold model, the greedy algorithm obtains a $(1-(1-1/k)^k)$-approximation on directed graphs in the independent cascade model.
However, \citet{khanna2014influence} showed that, in undirected graphs, the greedy algorithm performs substantially better: a $(1-(1-1/k)^k + c)$ approximation for constant $c > 0$.
Our results show that, surprisingly, no such improvement occurs in the linear threshold model.

Finally, we show that, under the linear threshold model, the approximation ratio $(1 - (1 - 1/k)^k)$ is tight if 1) the graph is directed or 2) the vertices are weighted. In other words, under either of these two settings, the greedy algorithm cannot achieve a $(1 - (1 - 1/k)^k + f(k))$-approximation for any positive function $f(k)$. The result in setting 2) is again in a sharp contrast to Khanna and Lucier’s $(1 - (1 - 1/k)^k + c)$-approximation result for the independent cascade model, where the $(1 - (1 - 1/k)^k + c)$ approximation guarantee can be extended to the setting where vertices are weighted.

We also discuss extensions to more generalized settings including those with edge-weighted graphs.
\end{abstract}

\section{Introduction}
Viral marketing is an advertising strategy that gives the company's product to a certain number of users (the seeds) for free such that the product can be promoted through a cascade process in which the product is recommended to these users' friends, their friends' friends, and so on.  The \emph{influence maximization problem} (\infmax) is an optimization problem which asks which seeds one should give the product to; that is, given a graph, a \emph{diffusion model} defining how each node is infected by its neighbors, and a limited budget $k$, how to pick $k$ seeds such that the total number of infected vertices in this graph at the end of the cascade is maximized.  For \infmax, nearly all the known algorithms are based on a greedy algorithm  which iteratively picks the seed that has the largest marginal influence.
Some of them improve the running time of the original greedy algorithm by skipping vertices that are known to be suboptimal~\cite{leskovec2007cost,goyal2011celf++}, while the others improve the scalability of the greedy algorithm by using more scalable algorithms to approximate the expected total influence~\cite{borgs2014maximizing,tang2014influence,tang2015influence,cheng2013staticgreedy,ohsaka2014fast} or computing a score of the seeds that is closely related to the expected total influence~\cite{chen2009efficient,chen2010scalable,ChenYZ10,goyal2011simpath,jung2012irie,galhotra2016holistic,schoenebeck2019influence}.
Therefore, improving the approximation guarantee of the standard greedy algorithm improves the approximation guarantees of most \infmax algorithms in the literature in one shot!

Two diffusion models that have been studied almost exclusively are \emph{the linear threshold model} and \emph{the independent cascade model}, which were proposed by~\citet{KempeKT03}.
In the independent cascade model, a newly-infected vertex (or seed) $u$ infects each of its not-yet-infected neighbors $v$ with a fixed probability independently.
In the linear threshold model for unweighted graphs\footnote{The linear threshold model can be defined for general weighted directed graphs. However, if the graph is undirected, the linear threshold model is normally defined with the edges unweighted. Since this paper mainly deals with undirected graphs, we will adopt the definition of the linear threshold model for unweighted graphs.}, each non-seed vertex has a threshold sampled uniformly and independently from the interval $[0, 1]$, and becomes infected when the fraction of its infected neighbors exceeds this threshold.

Both models were shown to be \emph{submodular} (see Theorem~\ref{thm:submodular} for details) even in the case with directed graphs~\cite{KempeKT03}, which implies that the greedy algorithm achieves a $(1-(1-1/k)^k)$-approximation for the \infmax problem, or, a $(1-1/e)$-approximation for any $k$.
A natural and important question is, can we show that the greedy algorithm can perform better than a $(1-(1-1/k)^k)$-approximation through a more careful analysis?

To answer this question, it is helpful to notice that \infmax is a special case of the \maxc problem: given a collection of subsets of a set of elements and a positive integer $k$, find $k$ subsets that cover maximum number of elements (see details in Sect.~\ref{sect:premaxc}).
For \maxc, it is well known that the greedy algorithm cannot overcome the $(1-(1-1/k)^k)$ barrier: for any positive function $f(k)$ which may be infinitesimal, there exists a \maxc instance where the greedy algorithm cannot achieve $(1-(1-1/k)^k+f(k))$-approximation.
Thus, to hope that the greedy algorithm can overcome this barrier for \infmax, we need to find out what makes \infmax more special and exploit those \infmax features that are not in \maxc.

Unfortunately, \infmax with the independent cascade model for general directed graphs is nothing more special than \maxc, as it can simulate any \maxc instance: set the probability that $u$ infects $v$ to be $1$ for all edges $(u,v)$ (i.e., a vertex will be infected if it contains an infected in-neighbor); use a vertex to represent a subset in the \maxc instance, and use a clique of size $m$ to represent an element; create a directed edge from the vertex representing the subset to an arbitrary vertex in the clique representing the element if this subset contains this element.
It is easy to see that this simulates a \maxc instance if $m$ is sufficiently large.
Therefore, the greedy algorithm cannot achieve a $(1-(1-1/k)^k+f(k))$-approximation for any positive function $f(k)$.  This implies we must use properties beyond mere submodularity (a property shared by \maxc) to improve the algorithmic analysis.

\citet{khanna2014influence} showed that the  $(1-(1-1/k)^k)$ barrier can be overcome if we restrict the graphs to be undirected in the independent cascade model.  They proved that the greedy algorithm for \infmax with the independent cascade model for undirected graphs achieves a $(1-(1-1/k)^k+c)$-approximation for some constant $c>0$ that does not even depend on $k$.\footnote{\citet{khanna2014influence} only claimed that the greedy algorithm achieves a $(1-1/e+c)$-approximation. However, $c$ being a constant implies that there exists $k_0$ such that $1-(1-1/k)^{k}<1-1/e+c/2$ for all $k\geq k_0$ (notice that $(1-(1-1/k)^k)$ is decreasing and approaches to $1-1/e$); the greedy algorithm will then achieve a $(1-(1-1/k)^k+c/2)$-approximation for $k\geq k_0$.
}
This means the greedy algorithm produces a $(1-1/e+c)$-approximation for any $k$.
Moreover, this result holds for the more general setting where 1) there is a prescribed set of vertices $V'\subseteq V$ as a part of input to the \infmax instance such that the seeds can only be chosen among vertices in $V'$ and 2) a positive weight is assigned to each vertex such that the objective is to maximize the total weight of infected vertices (instead of the total number of infected vertices).
This result is remarkable, as many of the social networks in our daily life are undirected by their nature (for example, friendship, co-authorship, etc.).
Knowing that the $(1-(1-1/k)^k)$ barrier can be overcome for the independent cascade model, a natural question is, what is the story for the linear threshold model?

\subsection{Our Results}
We show that Khanna and Lucier’s result on the independent cascade model can only be partially extended to the linear threshold model.
Our first result is an example showing that the greedy algorithm can obtain at most a $(1- (1 - 1/k)^k + O(1/k^{0.2}))$-approximation for \infmax on undirected graphs under the linear threshold model. This shows that, up to lower order terms, the approximation guarantee $1- (1 - 1/k)^k$ is tight.  In particular, no analogue of Khanna and Lucier’s $(1 - 1/e + c)$ result is possible if $c > 0$ is a constant.
For the greedy algorithm, we define \emph{the approximation surplus at $k$} be the additive term after $1-(1-1/k)^k$ in the approximation ratio.
Our result can then be equivalently stated as the approximation surplus at $k$ for the linear threshold model is $O(1/k^{0.2})$.

For our second result, we prove that the greedy algorithm does achieve a $(1 - (1 - 1/k)^k + \Omega(1/k^3))$-approximation under the same setting (the linear threshold model with undirected graphs). This indicates that the greedy algorithm can overcome the $(1 - (1 - 1/k)^k)$ barrier by a lower order term.
In particular, the barrier is overcome for constant $k$.
We remark that the approximation surplus $\Omega(1/k^3)$ does not depend on the number of vertices/edges in the graph, so this improvement is not diminishing as the size of the graph grows.

Finally, we extend our results to other \infmax settings.   Firstly, we show that the approximation ratio $(1 - (1 - 1/k)^k)$ is tight if we consider general directed graphs. That is, the greedy algorithm cannot achieve a $(1 - (1 - 1/k)^k + f(k))$-approximation for any positive function $f(k)$. Secondly, while still considering undirected graphs, we consider the two generalizations considered by~\citet{khanna2014influence}. We show that our result that the greedy algorithm achieves a $(1 - (1 - 1/k)^k + \Omega(1/k^3))$-approximation can be extended to the setting where the seeds can only be picked from a prescribed vertex set.  However, it cannot be extended to the setting where the vertices are weighted, in which case the approximation ratio of $(1 - (1 - 1/k)^k)$ is tight, as it is in directed graphs.
These results, as well as the corresponding result for the independent cascade model by \citet{khanna2014influence}, are summarized in Table~\ref{tab:generalize}.

\begin{table}
    \centering
    \begin{tabular}{|c|c|c|c|c|}
        \hline
            & \multicolumn{2}{c|}{Linear Threshold} & \multicolumn{2}{c|}{Independent Cascade}\\
        \hline
        Approximation & at least $\Omega(1/k^3))$  & less than $f(k)$  &  at least some  & less than $f(k)$ \\
        Surplus &  at most $O(1/k^{0.2}))$ &  for any $f(k)>0$ & constant $c>0$  &  for any $f(k)>0$\\
        \hline
        Directed & & \multirow{2}{*}{$\checkmark$} & & \multirow{2}{*}{$\checkmark$} \\
        Graph & & & & \\
        \hline
        Undirected &  \multirow{2}{*}{$\checkmark$} & & \multirow{2}{*}{$\checkmark$} & \\
        Graph & & & & \\
        \hline
        Undirected & & \multirow{4}{*}{$\checkmark$} &   \multirow{4}{*}{$\checkmark$} & \\
        Graph with & & & & \\
        Weighted & & & & \\
        Vertices & & & & \\
        \hline
        Undirected & \multirow{4}{*}{$\checkmark$} &  &   \multirow{4}{*}{$\checkmark$} & \\
        Graph with & & & & \\
        Prescribed & & & & \\
        Seed Set & & & & \\
        \hline
    \end{tabular}
    \caption{Approximation surplus of the greedy algorithm under different settings.}
    \label{tab:generalize}
\end{table}

We have defined the linear threshold model for \emph{unweighted}, undirected graphs where all the incoming edges of a vertex have the same weight.
We discuss alternative versions and extensions of the linear threshold model to edge-weighted graphs, and discuss how our results extend to these settings.

\subsection{Related Work}
The influence maximization problem was initially posed by Domingos and Richardson~\cite{DomingosR01,RichardsonD02}.
\citet{KempeKT03} showed
the linear threshold model and the independent cascade model are submodular, so the greedy algorithm achieves a $(1-(1-1/k)^k)$-approximation.
This result was later generalized to all diffusion models that are locally submodular~\cite{KempeKT05,MosselR10}.
As mentioned earlier, for the independent cascade model with undirected graphs, \citet{khanna2014influence} showed that the greedy algorithm achieves a $(1-(1-1/k)^k+c)$-approximation for some constant $c>0$.

On the hardness or inapproximability side, \citet{KempeKT03} showed that \infmax on both the linear threshold model and the independent cascade model is NP-hard.
For the independent cascade model with directed graphs, \citet{KempeKT03} showed a reduction from \maxc preserving the approximation factor.
Since \citet{feige1998threshold} showed that \maxc is NP-hard to approximated within factor $(1-(1-1/k)^k+\varepsilon)$ for any constant $\varepsilon>0$, the same inapproximability factor holds for the independent cascade \infmax.
Therefore, up to lower order terms, the gap between the upper bound and the lower bound for the independent cascade (on directed graphs) \infmax is closed.
If undirected graphs are considered, \citet{schoenebeck2019influence} showed that, for both the linear threshold model and the independent cascade model, \infmax is NP-hard to approximate to within factor $(1-\tau)$ for some constant $\tau>0$.

If the diffusion model can be nonsubmodular, \citet{KempeKT03} showed that \infmax is NP-hard to approximate to within a factor of $N^{1 - \varepsilon}$ for any $\varepsilon > 0$.
Many works after this \cite{chen2009approximability,li2017influence,schoenebeck2017beyond,toct2019beyond,Horel16} showed that strong inapproximability results extend to even very specific nonsubmodular models.

\infmax has also been studied in the adaptive setting, where the seeds are selected iteratively, and the seed-picker can observe the cascade of the previous seeds before choosing the next one~\cite{golovin2011adaptive,chen2019adaptivity,peng2019adaptive}.
Due to its iterative nature, the greedy algorithm can be easily generalized to an adaptive version~\cite{han2018efficient,chen2019adaptive}.

As mentioned in the introduction section, there was extensive work on designing implementations that are more efficient and scalable~\cite{leskovec2007cost,goyal2011celf++,borgs2014maximizing,tang2014influence,tang2015influence,cheng2013staticgreedy,ohsaka2014fast,chen2009efficient,chen2010scalable,goyal2011simpath,jung2012irie,galhotra2016holistic}.
These algorithms speedup the greedy algorithm by either disregarding those seed candidates that are identified to be clearly suboptimal or finding smart ways to approximate the expected number of infected vertices.
\citet{arora2017debunking} benchmark most of the aforementioned variants of the greedy algorithms.
We remark that there do exist \infmax algorithms that are not based on greedy~\cite{BKS07,GL13,angell2016don,schoenebeck2017beyond,toct2019beyond,schoenebeck2019think}, but they are typically for nonsubmodular diffusion models.

\section{Preliminaries}
\label{sect:pre}
\subsection{Influence Maximization with Linear Threshold Model}
\label{sect:preinfmax}
Throughout this paper, we use $G=(V,E)$ to represent the graph which may or may not be directed.
We use $S$ to denote the set of seeds,
$k$ to denote $|S|$.
Let $\deg(v)$ be the degree of $v$ when $G$ is undirected and the \emph{in-degree} of vertex $v$ otherwise.
For each $v\in V$, let $\Gamma(v)=\{u:(u,v)\in E\}$ be the set of \emph{(in-)neighbors} of vertex $v$.

\begin{definition}\label{def:LTM}
The \emph{linear threshold model} $LT_G$ is defined by a directed graph $G=(V,E)$.
On input seed set $S\subseteq V$, $LT_G(S)$ outputs a set of infected vertices as follows:
\begin{enumerate}
    \item Initially, only vertices in $S$ are infected, and for each vertex $v$ a \emph{threshold} $\theta_v\in\Z^+$ is sampled uniformly at random from $\{1,2,\ldots,\deg(v)\}$ independently. If $\deg(v)=0$, set $\theta_v=\infty$.
    \item In each subsequent iteration, a vertex $v$ becomes infected if $v$ has at least $\theta_v$ infected in-neighbors.
    \item After an iteration where there are no additional infected vertices, $LT_G(S)$ outputs the set of infected vertices.
\end{enumerate}
\end{definition}

In this paper, we mostly deal with undirected graphs.
When we restrict our attention to undirected graphs, the undirected graph is viewed as a special directed graph with each undirected edge of the graph being viewed as two anti-parallel directed edges.

Although the linear threshold model can be defined for general edge-weighted graphs, we will adopt the special case for the unweighted graphs as defined in Definition~\ref{def:LTM}, which is most common in the past literature when \emph{undirected graphs} are considered.
In particular, there are some subtle difficulties to define the linear threshold model on graphs that are both edge-weighted and undirected.
We discuss these in details in Append.~\ref{sect:generalizeLTM}.

Previous work showed that the linear threshold model has  \emph{live-edge interpretation} as stated in the theorem below.

\begin{theorem}[Claim~2.6 in~\cite{KempeKT03}]\label{thm:LT_live}
Let $\widehat{LT}_G(S)\subseteq V$ be the set of vertices that are reachable from $S$ when each vertex $v$ picks exactly one of its incoming edges uniformly at random to be included in the graph and vertices pick their incoming edges independently.
Then $\widehat{LT}_G(S)$ and $LT_G(S)$ have the same distribution.
Those picked edges are called ``live edges''.
\end{theorem}

The intuition of this interpretation is as follows: consider a not-yet-infected vertex $v$ and a set of its infected in-neighbors $\IN(v)\subseteq\Gamma(v)$.  By the definition of the linear threshold model, $v$ will be infected by vertices in $\IN(v)$ with probability $|\IN(v)|/\deg(v)$.  On the other hand, the live edge coming into $v$ will be from the set $\IN(v)$ with probability $|\IN(v)|/\deg(v)$.

Once again, when considering undirected graphs, those live edges in Theorem~\ref{thm:LT_live} are still directed.
Whenever we mention a live edge in the remaining part of this paper, it should always be clear that this edge is directed.

\begin{remark}\label{remark:uniquePath}
Since each vertex can choose only one incoming edge as being live, \emph{if a vertex $v$ is reachable from a vertex $u$ after sampling all the live edges, then there exists a unique simple path consisting of live edges connecting $u$ to $v$.}
\end{remark}

\begin{remark}\label{remark:RR}
When considering the probability that a given vertex $v$ will be infected by a given seed set $S$, we can consider a ``reverse random walk without repetition'' process.
The random walk starts at $v$, and it chooses one of its neighbors (in-neighbors for directed graphs) uniformly at random and moves to it. The random walk terminates when it reaches a vertex that has already been visited or when it reaches a seed.
Each move in the reverse random walk is analogous to selecting one incoming live edge. Theorem~\ref{thm:LT_live} implies that the probability that this random walk reaches a seed is exactly the probability that $v$ will be infected by seeds in $S$.
\end{remark}

Given a set of vertices $A$ and a vertex $v$, let $A\rightarrow v$ be the event that $v$ is reachable from $A$ after sampling live edges.
Alternatively, this means that the reverse random walk from $v$ described in Remark~\ref{remark:RR} reaches a vertex in $A$.
If $A$ is the set of seeds, then $\Pr(A\rightarrow v)$ is exactly the probability that $v$ will be infected.
Intuitively, $A\rightarrow v$ can be seen as the event that ``$A$ infects $v$''.
We set $\Pr(A\rightarrow v)=1$ if $v\in A$.
In this paper, we mean $A\rightarrow v$ when we say \emph{$v$ reversely walks to $A$} or \emph{$v$ is reachable from $A$}.
In particular, the reachability is in terms of the live edges, not the original edges.

Given a set of vertices $A$, a vertex $v$, and a set of vertices $B$, let $A\infect{B}v$ be the event that the reverse random walk from $v$ reaches a vertex in $A$ and the vertices on the live path from $v$ to $A$, excluding $v$ and the reached vertex in $A$, do not contain any vertex in $B$.
By definition, $A\infect{B}v$ is the same as $A\rightarrow v$ if $B=\emptyset$, and $\Pr(A\infect{B}v)=1$ for any $B$ if $v\in A$.

Let $\sigma(S)$ be the \emph{expected} total number of infected vertices due to the influence of $S$, $\sigma(S)=\E[|LT_G(S)|]$, where the expectation is taken over the samplings of thresholds of all vertices, or equivalently, over the choices of incoming live edges of all vertices.
By the linearity of expectation, we have $\sigma(S)=\sum_{v\in V}\Pr(S\rightarrow v)$.
It is known that computing $\sigma(S)$ or $\Pr(A\rightarrow v)$ for the linear threshold model is $\#$P-hard~\cite{chen2010scalable}.\footnote{Computing $\sigma(S)$ and $\Pr(S\rightarrow v)$ are also $\#$P-hard for the independent cascade model~\cite{chen2010scalable2}.}
On the other hand, a simple Monte Carlo sampling can approximate $\sigma(S)$ arbitrarily close with probability arbitrarily close to 1.
In this paper, we adopt the standard assumption $\sigma(\cdot)$ can be accessed by an oracle.

\begin{definition}
The \infmax problem is an optimization problem which takes as inputs $G=(V,E)$ and a positive integer $k$, and outputs
$\argmax_{S\subseteq V:|S|=k}\sigma(S)$,
a seed set of size $k$ that maximizes the expected number of infected vertices.
\end{definition}

The \emph{greedy algorithm} consists of $k$ iterations; in each iteration $i$, it includes the seed $s_i$ into the seed set $S$ (i.e., $S\leftarrow S\cup\{s_i\}$) with the highest marginal increment to $\sigma(\cdot)$: $s_i\in\argmax_{s\in V\setminus S}(\sigma(S\cup\{s\})-\sigma(S))$.
Under the linear threshold model, the objective function $\sigma(\cdot)$ is monotone and \emph{submodular} (see Theorem~\ref{thm:submodular}), which implies that the greedy algorithm achieves a $(1-(1-1/k)^k)$-approximation~\cite{Nemhauser78,KempeKT03}.
Notice that this approximation ratio becomes $1-1/e$ when $k$ tends to infinity, and $1-(1-1/k)^k>1-1/e$ for all positive $k$.

\begin{theorem}[\cite{KempeKT03}]\label{thm:submodular}
Consider \infmax with the linear threshold model. For any two sets of vertices $A,B$ with $A\subsetneq B$ and any vertex $v\notin B$, we have $\sigma(A\cup\{v\})-\sigma(A)\geq\sigma(B\cup\{v\})-\sigma(B)$,
and for any vertex $u\notin B\cup\{v\}$,
$\Pr\left(A\cup\{v\}\rightarrow u\right)-\Pr\left(A\rightarrow u\right)\geq\Pr\left(B\cup\{v\}\rightarrow u\right)-\Pr\left(B\rightarrow u\right)$.
\end{theorem}

Remark~\ref{remark:RR} straightforwardly implies the following lemma, which describes a negative correlation between the event that $\{u\}$ infects $v$ and the event that $u$ is infected by another seed set.
Some other properties for the linear threshold are presented in Sect.~\ref{sect:prooflowerbound}.
We introduce Lemma~\ref{lemma:negativecorrelation} in the preliminary section because this negative correlation property is a signature property that makes the linear threshold model quite different from the independent cascade model.
In the independent cascade model, knowing the existence of certain connections between vertices only makes it more likely that another pair of vertices are connected.
Intuitively, this is because, in the independent cascade model, each vertex does not ``choose'' one of its incoming edges, but rather, each incoming edge is included with a certain probability independently.
In addition, Lemma~\ref{lemma:negativecorrelation} holds for directed graphs, while all the lemmas in Sect.~\ref{sect:prooflowerbound} hold only for undirected graphs.

\begin{lemma}\label{lemma:negativecorrelation}
For any three sets of vertices $A,B_1,B_2$ and any two different vertices $u,v$, we have
$\Pr(A\infect{B_1}u)\geq\Pr(A\infect{B_1}u\mid \{u\}\infect{A\cup B_2}v)$.
\end{lemma}
\begin{proof}
Consider any simple path $p$ from $u$ to $v$.
If $u\infect{A\cup B_2}v$ happens with all edges in $p$ being live, then $\Pr(A\infect{B_1}u)\geq \Pr(A\infect{B_1}u\mid p\text{ is live})$.
This is apparent by noticing Remark~\ref{remark:RR}: if $p$ is already live, then the reverse random walk starting from $u$ should reach $A$ without touching any vertices on $p$ (if the random walk touches a vertex in $p$, it will follow the reverse direction of $p$ and eventually go back to $u$), which obviously happens with less probability compared to the case without restricting that the random walk cannot touch vertices on $p$.

Noticing this, the remaining part of the proof is trivial:
$$\Pr\left(A\infect{B_1}u\mid u\infect{A\cup B_2}v\right)=\sum_{p}\frac{\Pr(A\infect{B_1}u\mid p\mbox{ is live})\Pr(p\mbox{ is live})}{\Pr(u\infect{A\cup B_2}v)}$$
$$\qquad\leq\Pr(A\infect{B_1}u)\sum_p\frac{\Pr(p\mbox{ is live})}{\Pr(\{u\}\infect{A\cup B_2}v)}=\Pr\left(A\infect{B_1}u\right),$$
where the summation is over all simple paths $p$ connecting $u$ to $v$ without touching any vertices in $A\cup B_2$, and Remark~\ref{remark:uniquePath} ensures that the events ``$p$ is live'' over all possible such $p$'s form a partition of the event $u\infect{A\cup B_2}v$.
\end{proof}

\subsection{Influence Maximization Is A Special Case of \maxc}
\label{sect:premaxc}
In this section, we establish that linear threshold \infmax is a special case of the well-studied \maxc problem, a folklore that is widely known in the \infmax literature.
This section also introduces some key intuitions that will be used throughout the paper.
We will only discuss the linear threshold model for the purpose of this paper, although submodular \infmax in general can also be viewed as a special case of \maxc.

\begin{definition}
The \maxc problem is an optimization problem which takes as input a universe of elements $ U=\{e_1,\ldots,e_N\}$ , a collection of subsets $\M=\{S_1,\ldots,S_M:S_i\subseteq U\}$ and an positive integer $k$, and outputs a collection of $k$ subsets that maximizes the total number of covered elements: $\displaystyle\calS\in\argmax_{\calS\subseteq\M,|\calS|=k}\left|\bigcup_{S\in\calS}S\right|$.
Given $\calS\subseteq\M$, we denote $\displaystyle\val(\calS)=\left|\bigcup_{S\in\calS}S\right|$.
\end{definition}

It is well-known that the greedy algorithm (that iteratively selects a subset that maximizes the marginal increment of $\val(\cdot)$) achieves a $(1-(1-1/k)^k)$-approximation for \maxc (in Sect.~\ref{sect:propertiesmaxc}, we prove a more general statement stated in Lemma~\ref{lem:maxc_kl}).
On the other hand, this approximation guarantee is tight: for any positive function $f(k)>0$ which may be infinitesimal, there exists a \maxc instance such that the greedy algorithm cannot achieve a $(1-(1-1/k)^k+f(k))$-approximation.\footnote{Our result in Sect.~\ref{sect:generalize} says that the greedy algorithm cannot achieve a $(1-(1-1/k)^k+f(k))$-approximation for the linear threshold \infmax with \emph{directed} graphs, which provides a proof of this, since, as we will see soon, \infmax is a special case of \maxc.}
We will review some properties of \maxc in Sect.~\ref{sect:propertiesmaxc} that will be used in our analysis for \infmax.

\infmax with the linear threshold model can be viewed as a special case of \maxc in that an instance of \infmax can be transformed into an instance of \maxc.
Given an instance of \infmax $(G = (V, E), k)$, let $H$ be the set of all possible live-edge samplings.  That is, $H$ is the set of directed graphs on $V$ that are subgraphs of $G$ where each vertex has in-degree equal to 1.  In particular, $|H| = \prod_{v\in V}\deg(v)$.\footnote{Of course, vertices with in-degree $0$ should be excluded from this product. Whenever we write this product, we always refer to the one excluding vertices with in-degree $0$.} We create an instance of \maxc by letting the universe of elements be $V \times H$, i.e., pairs of vertices and live-edge samplings, $(v, g)$, where $v \in V$ and $g \in H$.  We then create a subset for each vertex $v \in V$. The subset corresponding to $v \in V$ contains $(u, g)$ if $u$ is reachable from $v$ in $g$.
Since
$\sigma(S)=\sum_{v\in V}\Pr(S\rightarrow v)=\sum_{v\in V}\frac{|\{g:\text{ }v\text{ is reachable from }S\text{ under }g\}|}{\prod_{w\in V}\deg(w)}=\frac{|\{(v,g):\text{ }v\text{ is reachable from }S\text{ under }g\}|}{\prod_{w\in V}\deg(w)}$, $\sigma(S)$ equals to the total number of elements covered by ``subsets'' in $S$, divided by $\prod_{v\in V}\deg(v)$.
As a result, $\sigma(S)$ is proportional to the total number of covered elements if viewing $S$ as a collection of subsets.
This establishes that \infmax is a special case of \maxc.
We denote by $\Sigma(S)=\{(u,g):u\mbox{ is reachable from }S\mbox{ under }g\}$ the set of ``elements'' that the ``subsets'' in $S$ cover, and we have $\sigma(S)=|\Sigma(S)|/\prod_{v\in V}\deg(v)$ as discussed above.

Having established the connection between \infmax and \maxc, we take a closer look at the intersection, union and difference of two subsets.
Let $S_1,S_2$ be two seed sets.
$\Sigma(S_1)\cup\Sigma(S_2)$ contains all those $(u,g)$ such that $u$ is reachable from either $S_1$ or $S_2$ under $g$.
Clearly, $\sigma(S_1\cup S_2)=|\Sigma(S_1\cup S_2)|/\prod_{v\in V}\deg(v)=|\Sigma(S_1)\cup\Sigma(S_2)|/\prod_{v\in V}\deg(v)$.
The first equality holds by definition which holds for set intersection and set difference as well.
The last equality, however, does not hold for set intersection and set difference.

$\Sigma(S_1)\cap\Sigma(S_2)$ contains all those $(u,g)$ such that $u$ is reachable from both $S_1$ and $S_2$ under $g$.
We have $|\Sigma(S_1)\cap\Sigma(S_2)|/\prod_{v\in V}\deg(v)=\sum_{v\in V}\Pr((S_1\rightarrow v)\land(S_2\rightarrow v))$.
For the special case where $S_1=\{u_1\}$ and $S_2=\{u_2\}$, by Remark~\ref{remark:uniquePath}, the event $(S_1\rightarrow v)\land(S_2\rightarrow v)$ can be partitioned into two disjoint events: 1) $v$ reaches $u_2$ before $u_1$ in the reverse random walk, $(\{u_1\}\infect{\{v\}} u_2)\land(\{u_2\}\infect{\{u_1\}}v)$, and 2) $v$ reaches $u_1$ before $u_2$ in the reverse random walk, $(\{u_2\}\infect{\{v\}} u_1)\land(\{u_1\}\infect{\{u_2\}}v)$.
For general $S_1$, $S_2$ with $S_1\cap S_2=\emptyset$, the event $(S_1\rightarrow v)\land(S_2\rightarrow v)$ can be partitioned into two disjoint events depending on whether $v$ reversely reaches $S_1$ or $S_2$ first.

Similarly, $\Sigma(S_1)\setminus\Sigma(S_2)$ contains all those $(u,g)$ such that $u$ is reachable from $S_1$ but not from $S_2$ under $g$, we have $|\Sigma(S_1)\setminus\Sigma(S_2)|/\prod_{v\in V}\deg(v)=\sum_{v\in V}\Pr((S_1\rightarrow v)\land\neg(S_2\rightarrow v))$.

\section{Upper Bound on Approximation Guarantee}
\label{sect:upperbound}
In this section, we show that the approximation guarantee for the greedy algorithm on \infmax is at most $(1- (1 - 1/k)^k + O(1/k^{0.2}))$ with the linear threshold model on undirected graphs.
In other words, the approximation surplus is $O(1/k^{0.2})$.
This shows that the approximation guarantee $(1-1/e)$ cannot be asymptotically improved, even if undirected graphs are considered.

Before we prove our main theorem in this section, we need the following lemma characterizing the cascade of a single seed on a complete graph which is interesting on its own.

\begin{lemma}\label{lem:clique}
Let $G$ be a complete graph with $n$ vertices, and let $S$ be a set containing a single vertex.
We have $\sigma(S)<3\sqrt{n}$.
\end{lemma}
The proof of Lemma~\ref{lem:clique} is in Appendix~\ref{append:proof_lem:clique}.
The intuition behind this lemma is simply the birthday paradox.
Consider the reverse random walk starting from any particular vertex $v$ with seed set $\{u\}$.  At each step, the walk chooses a random vertex other than the current vertex.  
By the birthday paradox, the expected time for the walk to reach a previously visited vertex is $\Theta(\sqrt{n})$.
The probability $v$ is infected is the probability that the random walk reaches the seed $\{u\}$ before reaching a previously visited vertex.  This is approximately $1-(1-1/n)^{\sqrt{n}}\approx1/\sqrt{n}$.
Finally, by the linearity of expectation, the total number of infected vertices is about $\sqrt{n}$.


The remainder of this section proves the following theorem.
\begin{theorem}\label{thm:upperbound}
Consider \infmax on undirected graphs with the linear threshold model. There exists an instance where the greedy algorithm only achieves a $(1- (1 - 1/k)^k + O(1/k^{0.2}))$-approximation.
\end{theorem}

The \infmax instance mentioned in Theorem~\ref{thm:upperbound} is shown below.

\begin{example}\label{example}
The example is illustrated in Fig.~\ref{fig:example}.
Given the number of seeds $k$, we construct the undirected graph $G=(V,E)$ with $k\lceil k^{1.2}\rceil+\lfloor(1-\frac{100}{k^{0.2}})k^{1.8}\rfloor$ vertices as follows.
Firstly, construct $k$ cliques $C_1,\ldots,C_k$ of size $\lceil k^{1.2}\rceil$, and in each clique $C_i$ label an arbitrary vertex $u_i$ .
Secondly, construct $k$ vertices $v_1,\ldots,v_k$.
For each $i=1,\ldots,k$, create $\lceil k^{0.8}(1-1/k)^{i-1}\rceil-1$ vertices and connect them to $v_i$.
For each $i$, those $\lceil k^{0.8}(1-1/k)^{i-1}\rceil-1$ vertices combined with $v_i$ form a star of size $\lceil k^{0.8}(1-1/k)^{i-1}\rceil$, and we will use $D_i$ to denote the $i$-th star.
Thirdly, we continue creating $\ell$ of these kinds of stars $D_{k+1},\ldots,D_{k+\ell}$ centered at $v_{k+1},\ldots,v_{k+\ell}$ such that $|D_{k+1}|=\cdots=|D_{k+\ell-1}|=\lceil k^{0.8}(1-1/k)^{k}\rceil, |D_{k+\ell}|\leq\lceil k^{0.8}(1-1/k)^{k}\rceil$, and $\sum_{i=1}^{k+\ell}|D_i|=\lfloor(1-\frac{100}{k^{0.2}})k^{1.8}\rfloor$. In other words, we keep creating stars of the same size $\lceil k^{0.8}(1-1/k)^{k}\rceil$ until we reach the point where the total number of vertices in all those stars is $\lfloor(1-\frac{100}{k^{0.2}})k^{1.8}\rfloor$ (we assume $k$ is sufficiently large), where the last star created may be ``partial'' and have a size smaller than $\lceil k^{0.8}(1-1/k)^{k}\rceil$.
Notice that $|D_1|\geq|D_2|\geq\cdots\geq|D_k|\geq|D_{k+1}|=\cdots=|D_{k+\ell-1}|\geq|D_{k+\ell}|=\Theta(k^{0.8})$.\footnote{These inequalities may not be strict. In fact, $|D_1|$ may be equal to $|D_2|$ as $k^{0.8}-k^{0.8}(1-1/k)=1/k^{0.2}<1$.}
Finally, create $k\times(k+\ell)$ edges $\{(u_i,v_j):i=1,\ldots,k;j=1,\ldots,k+\ell\}$.
\end{example}

\begin{figure}
    \centering
    \includegraphics[width=0.8\textwidth]{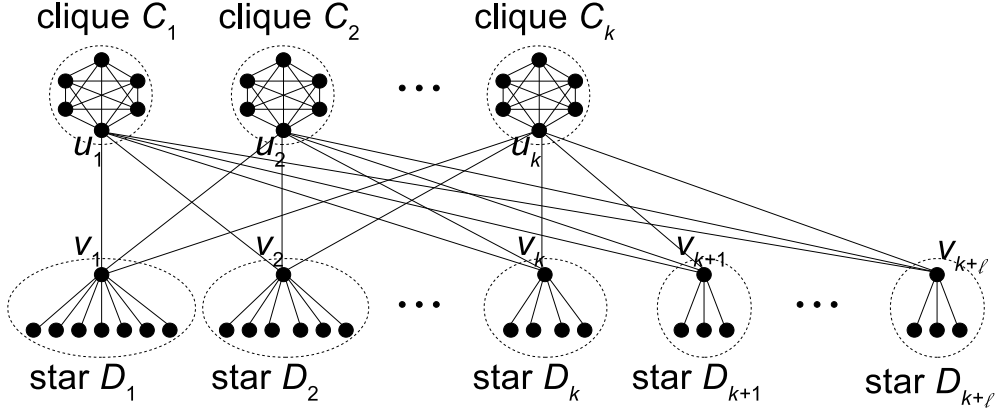}
    \caption{The tight example.}
    \label{fig:example}
\end{figure}


\paragraph{Proof Sketch of Theorem~\ref{thm:upperbound}}
We want that the greedy algorithm picks the seeds $v_1,\ldots,v_k$, while the optimal seeds are $u_1,\ldots,u_k$.
The purpose of constructing a clique $C_i$ for each $u_i$ is to simulate directed edges $(u_i,v_j)$ (such that, as mentioned earlier, each $u_i$ will be infected with $o(1)$ probability even if all of $v_1,\ldots,v_{k+\ell}$ are infected, and the total number of infections among the cliques is negligible so that the ``gadget'' itself is not ``heavy'').
In the optimal seeding strategy, each $v_i$ will be infected with probability $1-o(1)$, as the number of edges connecting to the seeds $u_1,\ldots,u_k$ is $k$, which is significantly more than the number of edges inside $D_i$ (which is at most $\lceil k^{0.8}\rceil$).
Therefore, $\sigma(\{u_1,\ldots,u_k\})\approx\sum_{i=1}^{k+\ell}|D_i|=\lfloor(1-\frac{100}{k^{0.2}})k^{1.8}\rfloor$, which is slightly less than $k^{1.8}$.
Moreover, each $\sigma(\{u_i\})$ is approximately $\frac1k$ of $\sigma(\{u_1,\ldots,u_k\})$, which is slightly less than $k^{0.8}$

The greedy algorithm would pick $v_1$ as the first seed, as $\sigma(v_1)$ is at least $\lceil k^{0.8}\rceil$ (by only accounting for the infected vertices in $D_1$) which is slightly larger than each $\sigma(\{u_i\})$.
After picking $v_1$ as the first seed, the marginal increment of $\sigma(\cdot)$ by choosing each of $u_1,\ldots,u_k$ becomes approximately $\frac1k\sum_{i=2}^{k+\ell}|D_i|=\frac1k(-|D_1|+\sum_{i=1}^{k+\ell}|D_i|)$, which is slightly less than $\frac1k(-\lceil k^{0.8}\rceil +k^{1.8})\approx |D_2|$.
On the other hand, noticing that $v_1$ infects each of $u_1,\ldots,u_k$ as well as $v_2$ with probability $o(1)$, the marginal increment of $\sigma(\cdot)$ by choosing $v_2$ is approximately $|D_2|$, which is slightly larger than the marginal increment by choosing any $u_i$ based on our calculation above.
Thus, the greedy algorithm will continue to pick $v_2$.
In general, we have designed the sizes of $D_1,D_2,\ldots,D_k$ such that they are just large enough to make sure the greedy algorithm will pick $v_1,v_2,\ldots,v_k$ one by one.

Our construction of cliques $C_1,\ldots,C_k$ makes sure that each of $u_1,\ldots,u_k$ will be infected with $o(1)$ probability even if all of $v_1,\ldots,v_k$ are seeded.
Therefore, $\sigma(\{v_1,\ldots,v_k\})\approx\sum_{i=1}^k|D_i|=\sum_{i=1}^k\lceil k^{0.8}(1-1/k)^{i-1}\rceil\leq k+\sum_{i=1}^kk^{0.8}(1-1/k)^{i-1}=k+k^{1.8}(1-(1-1/k)^k)$.
On the other hand, we have seen that $\sigma(\{u_1,\ldots,u_k\})$ is just slightly less than $k^{1.8}$.
To be more accurate, $\sigma(\{u_1,\ldots,u_k\})\approx (1-\frac{100}{k^{0.2}})k^{1.8}$.
Dividing $\sigma(\{v_1,\ldots,v_k\})$ by $\sigma(\{u_1,\ldots,u_k\})$ gives us the desired upper bound on the approximation ratio in Theorem~\ref{thm:upperbound}.
The numbers $0.2,0.8,1.2$ on the exponent of $k$ are optimized for getting the tightest bound while ensuring that the greedy algorithm still picks $v_1,\ldots,v_k$.

The remainder of this section aims to make the arguments above rigorous, and to derive the exact bound $(1- (1 - 1/k)^k + O(1/k^{0.2}))$.

Before we move on, we examine some of the properties of Example~\ref{example} which will be used later.

\begin{proposition}\label{prop:properties}
The followings are true.
\begin{enumerate}
    \item $\ell\leq k$;
    \item $\sigma(\{u_1\})=\cdots=\sigma(\{u_k\})$;
    \item $\sigma(\{v_1\})\geq\cdots\geq\sigma(\{v_{k+\ell}\})$;
    \item The greedy algorithm will never pick any vertices in $V\setminus\{u_1,\ldots,u_k,v_1,\ldots,v_k\}$;
    \item For any $i=1,\ldots,k$ and $j=1,\ldots,k+\ell$, we have $\Pr(u_i\rightarrow v_j)<\frac1k+\frac3{k^{1.2}}$;
    \item For any $i,j\in\{1,\ldots,k\}$ with $i\neq j$, we have $\Pr(u_i\rightarrow u_j)<\frac{2k}{k^{1.2}+2k-1}(\frac1k+\frac3{k^{1.2}})$.
\end{enumerate}
\end{proposition}
\begin{proof}
To show 1, suppose $\ell> k$, we will have
\begin{align*}
    \sum_{i=1}^{k+\ell}|D_i|&\geq\sum_{i=1}^k\left\lceil k^{0.8}\left(1-\frac1k\right)^{i-1}\right\rceil+\sum_{i=k+1}^{k+\ell-1}\left\lceil k^{0.8}\left(1-\frac1k\right)^{k}\right\rceil\\
    &\geq k^{1.8}\cdot \left(1-\left(1-\frac1k\right)^k\right)+k\cdot k^{0.8}\left(1-\frac1k\right)^{k}\\
    &=k^{1.8}>\left\lfloor\left(1-\frac{100}{k^{0.2}}\right)k^{1.8}\right\rfloor,
\end{align*}
which violates our construction.

2 follows immediately by symmetry, and 3 is trivial.
As for 4, choosing a seed in $D_i\setminus\{v_i\}$ is clearly sub-optimal, as choose $v_i$ as a seed will make all the remaining vertices in $D_i$ infected with probability $1$.
Choosing a seed in $C_i\setminus\{u_i\}$ is also sub-optimal.
If $u_i$ is not seeded, seeding $u_i$ is clearly better.
Otherwise, seeding any vertices from $u_1,\ldots,u_{i-1},u_{i+1},\ldots,u_k$ is better (notice that we have a total of $k$ seeds, so there are unseeded vertices among these).
This is due to the submodularity: if a clique $C_i$ already contains a seed, putting another seed in the same clique is no better than putting a seed in a new clique that does not contain a seed yet.
Finally, given that the greedy algorithm will choose seeds in $u_1,\ldots,u_k,v_1,\ldots,v_{k+\ell}$, choosing seeds in $\{v_{k+1},\ldots,v_{k+\ell}\}$ is clearly sub-optimal: we should first seed all of $v_1,\ldots,v_k$ before seeding any of $v_{k+1},\ldots,v_{k+\ell}$, but we have a total of only $k$ seeds.

To see 5, consider the reverse random walk starting from $v_j$.
Since $\deg(v_j)=k+|D_j|-1>k$, it will reach $u_i$ in one step with probability less than $1/k$.
It is easy to see that the walk will never reach $u_i$ if it ever reaches any vertex in $V\setminus\{u_1,\ldots,u_k,v_1,\ldots,v_{k+\ell}\}$.
Therefore, the only possibility of the walk reaching $u_i$ is to alternate between $\{u_1,\ldots,u_k\}$ and $\{v_1,\ldots,v_{k+\ell}\}$.
When it reaches a vertex on the $u$-side, it will move to a vertex on the $v$-side with probability $(k+\ell)/(k^{1.2}-1+k+\ell)$.
When it reaches a vertex on the $v$-side, it will move to exactly $u_i$ with probability less than $1/k$, as all vertices in $\{v_1,\ldots,v_{k+\ell}\}$ have degrees more than $k$.
If we disregard the scenario where the random walk visits a vertex that has already been visited (which can only increase the probability that the random walk reaches $u_i$), the random walk reaches $u_i$ at Step~3 with probability less than $\frac1k\cdot\frac{k+\ell}{k^{1.2}-1+k+\ell}$, it reaches $u_i$ at Step~5 with probability less than $\frac1k\cdot(\frac{k+\ell}{k^{1.2}-1+k+\ell})^2$, and so on.
Putting these analyses together,
$$\Pr(u_i\rightarrow v_j)<\sum_{t=0}^\infty\frac1k\cdot\left(\frac{k+\ell}{k^{1.2}-1+k+\ell}\right)^t=\frac1k\cdot\frac{k^{1.2}-1+k+\ell}{k^{1.2}-1}\leq \frac1k+\frac{2k}{k(k^{1.2}-1)}<\frac1k+\frac3{k^{1.2}},$$
where the penultimate inequality uses property 1.

To see 6, the reverse random walk starting from $u_j$ will reach the $v$-side with probability $(k+\ell)/(k^{1.2}-1+k+\ell)\leq2k/(k^{1.2}+2k-1)$ (since $\ell\leq k$ by 1).
Noticing this, property 5 and Lemma~\ref{lemma:negativecorrelation} conclude 6 immediately.
\end{proof}

\begin{proposition}
Given $G$ constructed in Example~\ref{example}, the greedy algorithm will iteratively pick $v_1,\ldots,v_k$.
\end{proposition}
\begin{proof}
By 4 in Proposition~\ref{prop:properties}, we will only consider seeds in $\{u_1,\ldots,u_k,v_1,\ldots,v_k\}$.
We will prove this proposition by induction.

For the base step, since choosing $v_1$ is more beneficial than choosing any of $v_2,\ldots,v_k$, we only need to compare $\sigma(\{v_1\})$ to each of $\sigma(\{u_1\}),\ldots,\sigma(\{u_k\})$.
Since $\sigma(\{u_1\})=\cdots=\sigma(\{u_k\})$, we consider $\sigma(\{u_1\})$ without loss of generality.
We aim to find an upper bound for $\sigma(\{u_1\})$ by upper-bounding the probability that each vertex in the graph is infected given a single seed $u_1$.

Firstly, the expected number of infected vertices in $C_1$ is at most $3k^{0.6}$ by Lemma~\ref{lem:clique}.
Next, by 6 in Proposition~\ref{prop:properties}, each of $u_2,\ldots,u_k$ will be infected with probability less than $\Pr(u_i\rightarrow u_j)<\frac{2k}{k^{1.2}+2k-1}(\frac1k+\frac3{k^{1.2}})$.
Moreover, if $u_i$ is not infected, all the remaining vertices in $C_i$ will not be infected.
If $u_i$ is infected, the total number of infected vertices in $C_i$ is at most $3k^{0.6}$ by Lemma~\ref{lem:clique}.
Finally, each vertex $v_1,\ldots,v_k$ will be infected with probability less than $\frac1k+\frac3{k^{1.2}}$ by 5 of Proposition~\ref{prop:properties}.
In addition, if certain $v_i$ is infected, then all vertices in $D_i$ will be infected.
Putting together, we have
\begin{align*}
    \sigma\left(\{u_1\}\right)&\leq 3k^{0.6}+(k-1)\cdot \frac{2k}{k^{1.2}+2k-1}\left(\frac1k+\frac3{k^{1.2}}\right)\cdot3k^{0.6}+\left(\frac1k+\frac3{k^{1.2}}\right)\left\lfloor\left(1-\frac{100}{k^{0.2}}\right)k^{1.8}\right\rfloor\tag{$\dag$}\\
    &\leq 3k^{0.6}+k\cdot\frac{2k}{k^{1.2}}\frac2k\cdot3k^{0.6}+\frac1k\left(1-\frac{100}{k^{0.2}}\right)k^{1.8}+\frac3{k^{1.2}}\left(1-\frac{100}{k^{0.2}}\right)k^{1.8}\\
    &<3k^{0.6}+12k^{0.4}+k^{0.8}-100k^{0.6}+3k^{0.6}\tag{$\ddag$}\\
    &<k^{0.8}.
\end{align*}

On the other hand, we have $\sigma(\{v_1\})\geq|D_1|=\lceil k^{0.8}\rceil>\sigma(\{u_1\})$.
Therefore, the first seed that the greedy algorithm will pick is $v_1$, which concludes the base step of the induction.

For the inductive step, suppose $v_1,\ldots,v_t$ have been chosen by the greedy algorithm in the first $t$ iterations.
We aim to show that the greedy algorithm will pick $v_{t+1}$ next.
By symmetry, with $v_1,\ldots,v_t$ being seeded, the marginal increment of $\sigma(\cdot)$ by seeding each of $u_1,\ldots,u_k$ is the same.
Thus, we only need to show that $\sigma(\{v_1,\ldots,v_{t+1}\})-\sigma(\{v_1,\ldots,v_t\})>\sigma(\{v_1,\ldots,v_{t},u_1\})-\sigma(\{v_1,\ldots,v_t\})$.

To calculate a lower bound for $\sigma(\{v_1,\ldots,v_{t+1}\})-\sigma(\{v_1,\ldots,v_t\})$, we first evaluate the probability $\Pr(\{v_1,\ldots,v_t\}\rightarrow v_{t+1})$.
In order for the reverse random walk starting from $v_{t+1}$ to reach one of $v_1,\ldots,v_t$, it must reach one of $u_1,\ldots,u_k$ in the first step, and then ``escape'' from the clique in the second step.
The probability that the walk escapes from the clique,  $(k+\ell)/(k^{1.2}-1+k+\ell)$, is clearly an upper bound of $\Pr(\{v_1,\ldots,v_t\}\rightarrow v_{t+1})$.
Therefore, with seeds $v_1,\ldots,v_t$, the expected number of infected vertices in $D_{t+1}$ is at most $(k+\ell)/(k^{1.2}-1+k+\ell)\times|D_{t+1}|$.
On the other hand, when $v_{t+1}$ is further seeded, all vertices in $D_{t+1}$ will be infected.
By only considering the marginal gain on the expected number of infected vertices in $D_{t+1}$, we have
$$\sigma(\{v_1,\ldots,v_{t+1}\})-\sigma(\{v_1,\ldots,v_t\})>\left(1-\frac{k+\ell}{k^{1.2}-1+k+\ell}\right)\cdot\left\lceil k^{0.8}\left(1-\frac1k\right)^t\right\rceil$$
$$\qquad>\left(1-\frac{2k}{k^{1.2}}\right)k^{0.8}\left(1-\frac1k\right)^t>k^{0.8}\left(1-\frac1k\right)^t-2k^{0.6}.$$

To find an upper bound for $\sigma(\{v_1,\ldots,v_{t},u_1\})-\sigma(\{v_1,\ldots,v_t\})$. We note that all vertices in $D_1,\ldots,D_t$ are infected with probability $1$ with seeds $v_1,\ldots,v_t$, and we have
\begin{align*}
    &\sigma(\{v_1,\ldots,v_{t},u_1\})-\sigma(\{v_1,\ldots,v_t\})\\ =&\sum_{w\in V\setminus(D_1\cup\cdots\cup D_t)}\left(\Pr(\{v_1,\ldots,v_{t},u_1\}\rightarrow w)-\Pr(\{v_1,\ldots,v_{t}\}\rightarrow w)\right)\\
    <&\sum_{w\in V\setminus(D_1\cup\cdots\cup D_t)}\Pr(\{u_1\}\rightarrow w)\tag{By Theorem~\ref{thm:submodular}},
\end{align*}
so we only need to consider the expected number of infected vertices with the graph containing only one seed $u_1$ and with vertices in $D_1\cup\cdots\cup D_t$ disregarded.

Therefore, if we split $\sigma(\{v_1,\ldots,v_{t},u_1\})-\sigma(\{v_1,\ldots,v_t\})$ into three terms as it is in ($\dag$), the first two terms regarding the expected number of infections on the $k$ cliques are the same as they appeared in ($\dag$), which are less than $3k^{0.6}+12k^{0.4}$ as computed at step ($\ddag$).
By excluding $D_1,\ldots,D_t$ for the third term, we have
\begin{align*}
    &\sigma(\{v_1,\ldots,v_{t},u_1\})-\sigma(\{v_1,\ldots,v_t\})\\
    <&3k^{0.6}+12k^{0.4}+\left(\frac1k+\frac3{k^{1.2}}\right)\left(\left\lfloor\left(1-\frac{100}{k^{0.2}}\right)k^{1.8}\right\rfloor-\sum_{i=1}^{t}\left\lceil k^{0.8}\left(1-\frac1k\right)^i\right\rceil\right)\\
    \leq&3k^{0.6}+12k^{0.4}+\left(\frac1k+\frac3{k^{1.2}}\right)\left(\left(1-\frac{100}{k^{0.2}}\right)k^{1.8}-\sum_{i=1}^{t} k^{0.8}\left(1-\frac1k\right)^i\right)\\
    \leq&3k^{0.6}+12k^{0.4}+k^{0.8}\left(1-\frac{100}{k^{0.2}}\right)-k^{0.8}\left(1-\left(1-\frac1k\right)^t\right)\\
    &+\frac3{k^{1.2}}\cdot \left(\left(1-\frac{100}{k^{0.2}}\right)k^{1.8}-\sum_{i=1}^{t} k^{0.8}\left(1-\frac1k\right)^i\right)\tag{since $\sum_{i=1}^{t} k^{0.8}\left(1-\frac1k\right)^i=k^{1.8}\left(1-\frac1k\right)^t$}\\
    \leq&3k^{0.6}+12k^{0.4}+k^{0.8}\left(1-\frac1k\right)^t-100k^{0.6}+\frac{3}{k^{1.2}}\cdot k^{1.8}\tag{since $\left(\left(1-\frac{100}{k^{0.2}}\right)k^{1.8}-\sum_{i=1}^{t} k^{0.8}\left(1-\frac1k\right)^i\right)\leq k^{1.8}$}\\
    <&k^{0.8}\left(1-\frac1k\right)^t-50k^{0.6}\\
    <&\sigma(\{v_1,\ldots,v_{t+1}\})-\sigma(\{v_1,\ldots,v_t\}),
\end{align*}
which concludes the inductive step.
\end{proof}

We are now ready to prove Theorem~\ref{thm:upperbound}.
Let $S=\{v_1,\ldots,v_k\}$ be the set of seeds selected by the greedy algorithm, and let $S^\ast=\{u_1,\ldots,u_k\}$. 

By only considering infected vertices in $D_1,\ldots,D_{k+\ell}$, we have
$$\sigma(S^\ast)>\frac{k}{k^{0.8}+k}\left\lfloor\left(1-\frac{100}{k^{0.2}}\right)k^{1.8}\right\rfloor>\frac{k}{k^{0.8}+k}\left(k^{1.8}-50k^{1.6}\right),$$
since each of $v_1,\ldots,v_{k+\ell}$ will be infected with probability at least $\frac{k}{k^{0.8}+k}$ (notice that even $v_1$, with the highest degree among $v_1,\ldots,v_{k+\ell}$, has degree only $k^{0.8}+k-1$).

Now consider $\sigma(S)$.
Given seed set $S$, each of $u_1,\ldots,u_k$ will be infected with probability $\frac{k}{k^{1.2}+k-1}$, and each of $v_{k+1},\ldots,v_{k+\ell}$ will be infected with probability at most $\frac{k}{k^{1.2}+k-1}$, as the reverse random walk starting from any of $v_{k+1},\ldots,v_{k+\ell}$ needs to reach one of $u_1,\ldots,u_k$ before reaching a seed in $S$.
Therefore,
\begin{align*}
    \sigma(S)&\leq \sum_{i=1}^{k}\left\lceil k^{0.8}\left(1-\frac1k\right)^i\right\rceil + k\cdot \frac{k}{k^{1.2}+k-1}\cdot3k^{0.6}+\ell\cdot \frac{k}{k^{1.2}+k-1}\left\lceil k^{0.8}\left(1-\frac1k\right)^k\right\rceil\\
    &\leq\sum_{i=1}^{k}k^{0.8}\left(1-\frac1k\right)^i+k+\frac{3k^{2.6}}{k^{1.2}}+k\cdot \frac{k}{k^{1.2}}k^{0.8}\tag{since $\lceil x\rceil\leq x+1$ and $\ell\leq k$}\\
    &=k^{1.8}\left(1-\left(1-\frac1k\right)^k\right)+k+3k^{1.4}+k^{1.6}.
\end{align*}

Finally, the approximation guarantee of the greedy algorithm on the instance described in Example~\ref{example} is at most
\begin{align*}
\frac{\sigma(S)}{\sigma(S^\ast)}&\leq\frac{k^{1.8}\left(1-\left(1-\frac1k\right)^k\right)+k+3k^{1.4}+k^{1.6}}{\frac{k}{k^{0.8}+k}\left(k^{1.8}-50k^{1.6}\right)}\\
&\leq \left(1-\left(1-\frac1k\right)^k\right)\frac{k^{1.8}(k^{0.8}+k)}{k(k^{1.8}-50k^{1.6})}+\frac{k+3k^{1.4}+k^{1.6}}{\frac12\times\frac12k^{1.8}}\tag{since $\frac{k}{k^{0.8}+k}>\frac12$ and $\frac12k^{1.8}\gg 50k^{1.6}$}\\
&=\left(1-\left(1-\frac1k\right)^k\right)\left(1+\frac{51k^{0.8}}{k-50k^{0.8}}\right)+\frac{4+12k^{0.4}+4k^{0.6}}{k^{0.8}}\\
&\leq\left(1-\left(1-\frac1k\right)^k\right)+O\left(\frac1{k^{0.2}}\right),
\end{align*}
which concludes Theorem~\ref{thm:upperbound}.

\section{Lower Bound on Approximation Guarantee}
\label{sect:lowerbound}
In this section, we prove that the greedy algorithm can obtain at least a $(1-(1-1/k)^k+\Omega(1/k^3))$-approximation to $\max_{S\subseteq V:|S|=k}\sigma(S)$, stated in Theorem~\ref{thm:lowerbound}.
That is, the approximation surplus is $\Omega(1/k^3)$.
This indicates that the barrier $1-(1-1/k)^k$ can be overcome if $k$ is a constant.
We have seen that \infmax is a special case of \maxc in Sect.~\ref{sect:premaxc}, and it is known that the greedy algorithm cannot overcome the barrier $1-(1-1/k)^k$ in \maxc.
Theorem~\ref{thm:lowerbound} shows that \infmax with the linear threshold model on undirected graphs has additional structure.
To prove Theorem~\ref{thm:lowerbound}, we first review in Sect.~\ref{sect:propertiesmaxc} some properties of \maxc that are useful to our analysis, and then we prove Theorem~\ref{thm:lowerbound} in Sect.~\ref{sect:prooflowerbound} by exploiting some special properties of \infmax that are not satisfied in \maxc.

\begin{theorem}\label{thm:lowerbound}
Consider \infmax on undirected graphs with the linear threshold model.
The greedy algorithm achieves a $(1-(1-1/k)^k+\Omega(1/k^3))$-approximation.
\end{theorem}

\subsection{Some Properties of \maxc}
\label{sect:propertiesmaxc}
In this section, we list some of the properties of \maxc which will be used in proving Theorem~\ref{thm:lowerbound}.
The proofs of the lemmas in this section are all standard, and are deferred to the appendix.
For all the lemmas in this section, we are considering a \maxc instance $(U,\M,k)$, where $\calS=\{S_1,\ldots,S_k\}$ denotes the $k$ subsets output by the greedy algorithm and $\calS^\ast=\{S_1^\ast,\ldots,S_k^\ast\}$ denotes the optimal solution.

\begin{restatable}{lemma}{maxcgo}\label{lem:maxc_g=o}
 If $S_1\in\calS^\ast$, then $\val(\calS)\geq(1-(1-\frac1k)^k+\frac1{4k^2})\val(\calS^\ast)$.
\end{restatable}

\begin{restatable}{lemma}{maxcfirstintersection}\label{lem:maxc_firstintersection}
If $\frac{|S_1\cap(\bigcup_{i=1}^kS_i^\ast)|}{\val(\calS^\ast)}\notin[\frac1k-\varepsilon,\frac1k+\varepsilon]$ for some $\varepsilon>0$ which may depend on $k$, then $\val(\calS)\geq(1-(1-1/k)^k+\varepsilon/4)\val(\calS^\ast)$.
\end{restatable}

\begin{restatable}{lemma}{maxcdisjoint}\label{lem:maxc_disjoint}
If $\sum_{i=1}^k|S_i^\ast|>(1+\varepsilon)\val(\calS^\ast)$ for some $\varepsilon>0$ which may depend on $k$, then $\val(\calS)\geq(1-(1-\frac1k)^k+\frac\varepsilon{8k})\val(\calS^\ast)$.
\end{restatable}

\begin{restatable}{lemma}{maxcoutside}\label{lem:maxc_outside}
If $|S_1\setminus(\bigcup_{i=1}^kS_i^\ast)|>\varepsilon\val(\calS^\ast)$ for some $\varepsilon>0$ which may depend on $k$, then $\val(\calS)\geq(1-(1-1/k)^k+\varepsilon/16)\val(\calS^\ast)$.
\end{restatable}

\begin{restatable}{lemma}{maxcbalance}\label{lem:maxc_balance}
If there exists $S_i^\ast\in\calS^\ast$ such that $|S_i^\ast|<(\frac1k-\varepsilon)\val(\calS^\ast)$ for some $\varepsilon>0$ which may depend on $k$, then $\val(\calS)\geq(1-(1-\frac1k)^k+\frac\varepsilon{8k})\val(\calS^\ast)$.
\end{restatable}

\subsection{Proof of Theorem~\ref{thm:lowerbound}}
\label{sect:prooflowerbound}
We begin by proving some properties that are exclusively for \infmax.

\begin{lemma}\label{lem:|A|1}
Given a subset of vertices $A\subseteq V$, a vertex $v\notin A$ and a neighbor $u\in\Gamma(v)$ of $v$, with probability at most $\frac{|A|}{|A|+1}$, there is a simple live path from a vertex in $A$ to vertex $v$ such that the last vertex in the path before reaching $v$ is not $u$.
\end{lemma}
\begin{proof}
We consider all possible reverse random walks starting from $v$, and define a mapping from those walks that eventually reach $A$ to those that do not.
For each reverse random walk that reaches a vertex $a\in A$, $v\leftarrow w_1\leftarrow\cdots\leftarrow w_{\ell-1}\leftarrow w_\ell\leftarrow a$ (with $w_1,\ldots,w_\ell\notin A$), we map it to the random walk $v\leftarrow w_1\leftarrow\cdots\leftarrow w_{\ell-1}\leftarrow w_\ell\leftarrow w_{\ell-1}$, i.e., the one with the last step moving back.
Notice that the latter reverse random walk visits $w_{\ell-1}$ more than once, and thus will not reach $A$.
Specifically, for those reverse random walks that reach $A$ in one single step $v\leftarrow a$ (in the case $v$ is adjacent to $a\in A$), we map it to the reverse random walk $v\leftarrow u$, which are excluded from the event that ``there is a simple live path from a vertex in $A$ to vertex $v$ such that the last vertex in the path before reaching $v$ is not $u$'' (if $v\leftarrow u$, then every path that reaches $v$ should then reach $u$ in the penultimate step).

It is easy to see that at most $|A|$ different reverse random walks that reach $A$ can be mapped to a same random walk that does not reach $A$.
In order to make different reverse random walks have the same image in the mapping, they must share the same path $v\leftarrow w_1\leftarrow\cdots\leftarrow w_\ell$ except for the last step.
The last step, which moves to a vertex in $A$, can only have $|A|$ different choices.
For the special reverse random walks that move to $A$ in one step, there are at most $|A|$ of them, which are mapped to the random walk $v\leftarrow u$.

It is also easy to see that each random walk happens with the same probability as its image does.
This is because $w_\ell$ chooses its incoming edges uniformly, so choosing $a$ happens with the same chance as choosing $w_\ell$.
Specifically, $v$ chooses its incoming edge $(a,v)$ with the same probability as $(u,v)$.

Since we have defined a mapping that maps at most $|A|$ disjoint sub-events in the positive case to a sub-event in the negative case with the same probability, the lemma follows.
\end{proof}

\begin{lemma}\label{lem:|A|2}
Given a subset of vertices $A\subseteq V$ and two different vertices $u,v\notin A$, we have $\Pr(A\rightarrow u\mid\{u\}\infect{A}v)\leq\frac{|A|}{|A|+1}$.
\end{lemma}
\begin{proof}
Let $w_1,\ldots,w_t$ enumerate all the neighbors of $u$ that are not in $A$.
For each $i=1,\ldots,t$, let $E_i$ be the event that the reverse random walk starting from $v$ reaches $u$ without touching $A$ and its last step before reaching $u$ is at $w_i$.
Clearly, $\{E_1,\ldots,E_t\}$ is a partition of $\{u\}\infect{A}v$.
Conditioning on the event $E_i$, if $A\rightarrow u$ happens, the reverse random walk from $u$ to $A$ cannot touch $w_i$, since $w_i$ has already chosen its incoming edge $(u,w_i)$ in the case $E_i$ happens.
Therefore, by Lemma~\ref{lemma:negativecorrelation} and Lemma~\ref{lem:|A|1}, $\Pr(A\rightarrow u\mid E_i)=\Pr(A\infect{\{w_i\}}u\mid E_i)\leq\Pr(A\infect{\{w_i\}}u)\leq\frac{|A|}{|A|+1}$.\footnote{Rigorously speaking, the statement of Lemma~\ref{lemma:negativecorrelation} does not directly imply $\Pr(A\infect{\{w_i\}}u\mid E_i)\leq\Pr(A\infect{\{w_i\}}u)$. However, the proof of Lemma~\ref{lemma:negativecorrelation} can be adapted to show this. Instead of summing over all simple paths $p$ from $u$ to $v$ in the summation of the last inequality in the proof, we sum over all simple paths from $u$ to $v$ \emph{such that $u$ first moves to $w_i$}. The remaining part of the proof is the same. The idea here is that, the event $v$ reversely walks to $u$ is negatively correlated to the event that $u$ reversely walks to $A$, as the latter walk cannot hit the vertices on the path $u\rightarrow v$ if there is already a path from $u$ to $v$.}
We have
$$\Pr(A\rightarrow u\mid\{u\}\infect{A}v)=\frac{\sum_{i=1}^t\Pr(A\rightarrow u\mid E_i)\Pr(E_i)}{\Pr(\{u\}\infect{A}v)}\leq\frac{|A|}{|A|+1}\frac{\sum_{i=1}^t\Pr(E_i)}{\Pr(\{u\}\infect{A}v)}=\frac{|A|}{|A|+1},$$
which concludes this lemma.
\end{proof}

Finally, we need the following lemma which is due to \citet{lim2015simple}, while a more generalized version is proved by \citet{schoenebeck2019influence}.

\begin{lemma}[\citet{lim2015simple}]\label{lem:deg+1}
For any $v\in V$, we have $\sigma(\{v\})\leq\deg(v)+1$.
\end{lemma}

A proof of a more generalized version of the lemma above, which extends this lemma to the linear threshold model \emph{with slackness} (see Append.~\ref{sect:generalizeLTM} for definition of this model), is included in Appendix~\ref{append:extension_slackness} for completeness.
The proof is mostly identical to the proof by \citet{schoenebeck2019influence}.

Now we are ready to show Theorem~\ref{thm:lowerbound}.
In the remaining part of this section, we use $S=\{v_1,\ldots,v_k\}$ and $S^\ast=\{u_1,\ldots,u_k\}$ to denote the seed sets output by the greedy algorithm and the optimal seed set respectively.
Recall that we have established that \infmax is a special case of \maxc in Sect.~\ref{sect:premaxc}, and $v_1,\ldots,v_k,u_1,\ldots,u_k$ can be viewed as subsets in \maxc.
Thus, the lemmas in Sect.~\ref{sect:propertiesmaxc} can be applied here.

First of all, if $v_1\in S^\ast$, Lemma~\ref{lem:maxc_g=o} implies Theorem~\ref{thm:lowerbound} already.
In particular, Lemma~\ref{lem:maxc_g=o} implies that $|\Sigma(S)|\geq(1-(1-1/k)^k+1/4k^2)|\Sigma(S^\ast)|$ (refer to Sect.~\ref{sect:premaxc} for the definition of $\Sigma(\cdot)$), which implies $\sigma(S)\geq(1-(1-1/k)^k+1/4k^2)\sigma(S^\ast)$ by dividing $\prod_{w\in V}\deg(w)$ on both side of the inequality.
Therefore, we assume $v_1\notin S^\ast$ from now on.

Next, we analyze the intersection between $\Sigma(\{v_1\})$ and $\Sigma(S^\ast)$.
As an overview of the remaining part of our proof, suppose the barrier $1-(1-1/k)^k$ cannot be overcome, Lemma~\ref{lem:maxc_disjoint} and Lemma~\ref{lem:maxc_balance} imply that $\Sigma(\{u_1\}),\ldots,\Sigma(\{u_k\})$ must be almost disjoint and almost balanced, Lemma~\ref{lem:maxc_firstintersection} implies that $\Sigma(\{v_1\})$ must intersect approximately $1/k$ fraction of $\Sigma(S^\ast)$, and Lemma~\ref{lem:maxc_outside} implies that $\Sigma(\{v_1\})\setminus\Sigma(S^\ast)$ should not be large.
We will prove that these conditions cannot be satisfied at the same time.

The intersection $\Sigma(\{v_1\})\cap\Sigma(S^\ast)$ consists of all the tuples $(w,g)$ such that $w$ is reachable from both $v_1$ and $S^\ast$ under the live-edge realization $g$.
Consider the reverse random walk starting from $w$.
There are three different disjoint cases: 1) $w$ reaches $v_1$ first, and then reaches a vertex in $S^\ast$; 2) $w$ reaches a vertex in $S^\ast$, and then reaches $v_1$; 3) $w$ visits more than one vertex in $S^\ast$, and then reaches $v_1$.
The three terms in the following equation, which are named $C_1,C_2,C_3$, correspond to these three cases respectively.
\begin{align*}
    \frac{|\Sigma(\{v_1\})\cap\Sigma(S^\ast)|}{\prod_{w\in V}\deg(w)}=&\sum_{w\in V}\Pr\left( \left(S^\ast\rightarrow v_1\right)\land\left(\{v_1\}\infect{S^\ast}w\right)\right)\tag{$C_1$}\\
    +&\sum_{w\in V}\sum_{i=1}^k\Pr\left( \left(\{v_1\}\infect{S^\ast}u_i\right)\land \left(\{u_i\}\infect{S^\ast}w\right)\right)\tag{$C_2$}\\
    +&\sum_{w\in V}\sum_{i\neq j}\Pr\left(\left(\{v_1\}\rightarrow u_j\right)\land\left(\{u_j\}\infect{S^\ast}u_i\right)\land\left(\{u_i\}\infect{S^\ast}w\right)\right)\tag{$C_3$}
\end{align*}
Notice that this decomposition assumes $v_1\notin S^\ast$.

Firstly, we show that $C_1$ cannot be too large if the barrier $1-(1-1/k)^k$ is not overcome.
Intuitively, $C_1$ describes those $w$ that first reversely reaches $v_1$ and then reversely reaches a vertex in $S^\ast$.
Lemma~\ref{lem:|A|2} tells us that $v_1$ will reversely reach $S^\ast$ with at most probability $k/(k+1)$ conditioning on $w$ reversely reaching $v_1$.
This implies that, if $w$ reversely reaches $v_1$, $v_1$ will not reversely reach $S^\ast$ with probability at least $1/(k+1)$, which is at least $1/k$ of the probability that $v_1$ reversely reaches $S^\ast$.
Therefore, whenever we have a certain number of elements in $\Sigma(\{v_1\})\cap\Sigma(S^\ast)$ that corresponds to $C_1$, we have at least $1/k$ fraction of this number in $\Sigma(\{v_1\})\setminus\Sigma(S^\ast)$.
Lemma~\ref{lem:maxc_outside} implies that the $1-(1-1/k)^k$ barrier can be overcome if $|\Sigma(\{v_1\})\setminus\Sigma(S^\ast)|$ is large.
\begin{proposition}\label{prop:C1}
If $C_1>\frac{9}{10k}\cdot\sigma(S^\ast)$, then $\sigma(S)\geq(1-(1-\frac1k)^k+\frac1{640k^2})\cdot\sigma(S^\ast)$.
\end{proposition}
\begin{proof}
If $w=v_1$, $\{v_1\}\infect{S^\ast}w$ happens automatically, and $\Pr((\{v_1\}\infect{S^\ast}w)\land(S^\ast\rightarrow v_1))=\Pr(S^\ast\rightarrow v_1)$.
Substituting this into $C_1$, we have
\begin{align*}
    C_1&=\Pr(S^\ast\rightarrow v_1)+\sum_{w\in V\setminus \{v_1\}}\Pr\left(\left(S^\ast\rightarrow v_1\right)\land\left(\{v_1\}\infect{S^\ast}w\right)\right)\\
    &\leq1+\sum_{w\in V\setminus \{v_1\}}\Pr\left(\{v_1\}\infect{S^\ast}w\right)\cdot\Pr\left(S^\ast\rightarrow v_1\mid \{v_1\}\infect{S^\ast}w\right)\\
    &\leq1+\sum_{w\in V\setminus \{v_1\}}\Pr\left(\{v_1\}\infect{S^\ast}w\right)\cdot k\Pr\left(\neg(S^\ast\rightarrow v_1)\mid \{v_1\}\infect{S^\ast}w\right)\tag{Lemma~\ref{lem:|A|2}}\\
    &=1+k\sum_{w\in V\setminus\{v_1\}}\Pr\left(\left(\{v_1\}\infect{S^\ast}w\right)\land \neg\left(S^\ast\rightarrow v_1\right)\right),
\end{align*}
where the penultimate step is due to Lemma~\ref{lem:|A|2} from which we have $\Pr(S^\ast\rightarrow v_1\mid \{v_1\}\infect{S^\ast}w)\leq\frac{k}{k+1}$, which implies $\Pr(\neg(S^\ast\rightarrow v_1)\mid \{v_1\}\infect{S^\ast}w)\geq\frac1{k+1}$, which further implies $\Pr(S^\ast\rightarrow v_1\mid \{v_1\}\infect{S^\ast}w)\leq k\cdot \Pr(\neg(S^\ast\rightarrow v_1)\mid \{v_1\}\infect{S^\ast}w)$.

Notice that the summation $\sum_{w\in V\setminus\{v_1\}}\Pr((\{v_1\}\infect{S^\ast}w)\land \neg(S^\ast\rightarrow v_1))$ describes those $(w,g)$ such that $w$ is reachable from $v_1$ but not $S^\ast$ under realization $g$, which corresponds to elements in $\Sigma(\{v_1\})\setminus\Sigma(S^\ast)$.
Therefore, we have
$$\frac{|\Sigma(\{v_1\})\setminus\Sigma(S^\ast)|}{\prod_{w\in V}\deg(w)}\geq\sum_{w\in V\setminus\{v_1\}}\Pr\left(\left(\{v_1\}\infect{S^\ast}w\right)\land \neg\left(S^\ast\rightarrow v_1\right)\right)\geq\frac{C_1-1}k.$$

If $\sigma(S^\ast)\leq\frac87k$, we can see that $\sigma(S)\geq k\geq\frac78\sigma(S^\ast)>(1-(1-\frac1k)^k+\frac1{640k^2})\sigma(S^\ast)$ and the proposition is already implied.
Thus, we assume $\sigma(S^\ast)>\frac87k$ from now on.

If we have $C_1>\frac9{10k}\sigma(S^\ast)$ as given in the proposition statement, we have $C_1-1>\frac9{10k}\sigma(S^\ast)-\frac{7}{8k}\sigma(S^\ast)=\frac1{40k}\sigma(S^\ast)=\frac1{40k}\frac{|\Sigma(S^\ast)|}{\prod_{w\in V}\deg(w)}$.
Putting together,
$$\frac{|\Sigma(\{v_1\})\setminus\Sigma(S^\ast)|}{\prod_{w\in V}\deg(w)}\geq\frac{C_1-1}k>\frac1{40k^2}\frac{|\Sigma(S^\ast)|}{\prod_{w\in V}\deg(w)},$$
which yields $|\Sigma(\{v_1\})\setminus\Sigma(S^\ast)|>\frac1{40k^2}|\Sigma(S^\ast)|$.
Lemma~\ref{lem:maxc_outside} implies $|\Sigma(S)|\geq(1-(1-\frac1k)^k+\frac1{640k^2})|\Sigma(S^\ast)|$, which further implies this proposition.
\end{proof}

Secondly, we show that $C_2$ cannot be too large if the barrier $1-(1-1/k)^k$ is not overcome.
To show this, we first show that there exists $u_i\in S^\ast$ such that $\Pr(\{v_1\}\rightarrow u_i)\geq\frac{C_2}{\sigma(S^\ast)}$, and then show that this implies that $|\Sigma(\{v_1\})\setminus\Sigma(S^\ast)|$ is large by accounting for $v_1$'s influence to $u_i$'s neighbors.

\begin{proposition}\label{prop:C2}
If $C_2>\frac1{100k}\cdot\sigma(S^\ast)$, then $\sigma(S)\geq(1-(1-\frac1k)^k+\frac1{64000k^3})\sigma(S^\ast)$.
\end{proposition}
\begin{proof}
We give an outline of the proof first.
Assume $\displaystyle u_1\in\argmax_{u_i\in S^\ast}\Pr\left(\{v_1\}\infect{S^\ast}u_i\right)$ without loss of generality.
The proof is split into two steps.
\begin{itemize}
    \item Step 1: We will show that $\sum_{w\in\Gamma(u_1)\setminus S^\ast}\Pr(\{v_1\}\infect{S^\ast}w)=\Omega\left(\frac1{k^2}\right)\sigma(S^\ast)$ if we have $C_2>\frac1{100k}\cdot\sigma(S^\ast)$ in the proposition statement.
    Notice that the summation consists of the neighbors of $u_1$ (that are not in $S^\ast$) that reversely reaches $v_1$, which is a lower bound to $\sigma(v_1)$ ($v_1$ may infect more vertices than only the neighbors of $u_1$).
    To show this, we first find an upper bound of $C_2$ in terms of this summation:
    $\frac{C_2}{\sigma(S^\ast)}\leq \frac1{\deg(u_1)}\sum_{w\in\Gamma(u_1)\setminus S^\ast}\Pr(\{v_1\}\infect{S^\ast}w)$.
    This will imply that $\sum_{w\in\Gamma(u_1)\setminus S^\ast}\Pr(\{v_1\}\infect{S^\ast}w)=\Omega\left(\frac1{k^2}\right)\sigma(S^\ast)$ if assuming $C_2>\frac1{100k}\cdot\sigma(S^\ast)$, because $\deg(u_1)$ is (approximately) an upper bound to $\sigma(\{u_1\})$ by Lemma~\ref{lem:deg+1}, and $\sigma(\{u_1\})$ is approximately $\frac1k\sigma(S^\ast)$ (otherwise, the proposition holds directed by Lemma~\ref{lem:maxc_balance}).
    \item Step 2: We will show that $\Pr(\neg(S^\ast\rightarrow v_1)\mid\{v_1\}\infect{S^\ast}w)\geq\frac1{2(k+1)}$ for each $w\in\Gamma(u_1)\setminus S^\ast$.
    This says that, for each of $u_1$'s neighbor $w$, if it reversely reaches $v_1$, it will not reach $S^\ast$ with a reasonably high probability. Correspondingly, a reasonably large fraction of $\Sigma(\{v_1\})$ will not be in $\Sigma(S^\ast)$. By Lemma~\ref{lem:maxc_outside}, this proposition is concluded.
\end{itemize}

\paragraph{Step 1}
Based on the first vertex in $S^\ast$ that $w$ reversely reaches, we can decompose $\sigma(S^\ast)$ as follows:
$$\sigma(S^\ast)=\sum_{w\in V}\sum_{i=1}^k\Pr\left(\{u_i\}\infect{S^\ast}w\right).$$
Next, we have
\begin{align*}
    \frac{C_2}{\sigma(S^\ast)}&=\frac{\sum_{w\in V}\sum_{i=1}^k\Pr\left(\{u_i\}\infect{S^\ast}w\right)\Pr\left(\{v_1\}\infect{S^\ast}u_i\mid\{u_i\}\infect{S^\ast}w \right)}{\sum_{w\in V}\sum_{i=1}^k\Pr\left(\{u_i\}\infect{S^\ast}w\right)}\\
    &\leq\frac{\sum_{w\in V}\sum_{i=1}^k\Pr\left(\{u_i\}\infect{S^\ast}w\right)\Pr\left(\{v_1\}\infect{S^\ast}u_i\right)}{\sum_{w\in V}\sum_{i=1}^k\Pr\left(\{u_i\}\infect{S^\ast}w\right)}\tag{Lemma~\ref{lemma:negativecorrelation}}\\
    &\leq\Pr\left(\{v_1\}\infect{S^\ast}u_1\right)\cdot \frac{\sum_{w\in V}\sum_{i=1}^k\Pr\left(\{u_i\}\infect{S^\ast}w\right)}{\sum_{w\in V}\sum_{i=1}^k\Pr\left(\{u_i\}\infect{S^\ast}w\right)}\\
    &=\Pr\left(\{v_1\}\infect{S^\ast}u_1\right)\\
    &=\frac1{\deg(u_1)}\sum_{w\in\Gamma(u_1)\setminus S^\ast}\Pr\left(\{v_1\}\infect{S^\ast}w\right).
\end{align*}
For the last step, $v_1$ needs to first connect to one of $u_1$'s neighbors before connecting to $u_1$.
Notice that these neighbors may include $v_1$ itself.
In this special case $w=v_1\in\Gamma(u_1)\setminus S^\ast$, we have $\Pr(\{v_1\}\infect{S^\ast}w)=1$ and $u_1$ chooses its incoming live edge to be $(v_1,u_1)$ with probability $\frac1{\deg(u_1)}$, which is also a valid term in the summation above.

If $C_2>\frac1{100k}\cdot\sigma(S^\ast)$ as suggested by the proposition statement, we have
$$\sum_{w\in\Gamma(u_1)\setminus S^\ast}\Pr\left(\{v_1\}\infect{S^\ast}w\right)\geq\frac{\deg(u_1)C_2}{\sigma(S^\ast)}>\frac{\deg(u_1)}{100k}\geq\frac{\deg(u_1)+1}{200k}\geq\frac{\sigma(\{u_1\})}{200k}\geq \frac{9\sigma(S^\ast)}{2000k^2},$$
where the penultimate step is due to Lemma~\ref{lem:deg+1} and the last step is based on the assumption $\sigma(\{u_1\})\geq\frac{9}{10k}\sigma(S^\ast)$.
Notice that we can assume this without loss of generality, as otherwise Lemma~\ref{lem:maxc_balance} implies that $|\Sigma(S)|\geq(1-(1-\frac1k)^k+\frac{1}{80k^2})|\Sigma(S^\ast)|$, which directly implies this proposition.

\paragraph{Step 2}
If $w\neq v_1$, Lemma~\ref{lem:|A|2} implies that $\Pr(\neg(S^\ast\rightarrow v_1)\mid\{v_1\}\infect{S^\ast}w)\geq\frac1{k+1}>\frac1{2(k+1)}$.
If $w=v_1$, then $u_1$ and $v_1$ are adjacent.
Notice that $\deg(v_1)\geq 2$, for otherwise $\sigma(\{u_1\})>\sigma(\{v_1\})$ so $v_1$ cannot be the first seed picked by the greedy algorithm.
Therefore, $v_1$ reversely reaches $u_1$ in one step with probability at most $\frac12$.
If $v_1$ reversely reaches a vertex in $S^\ast$ such that the first step of the reverse random walk is not towards $u_1$, Lemma~\ref{lem:|A|1} implies that the probability this happens is at most $\frac{k}{k+1}$.
Putting together, for $w=v_1$, $\Pr(S^\ast\rightarrow v_1\mid\{v_1\}\infect{S^\ast}w)\leq \frac12+\frac12\cdot\frac{k}{k+1}$.
Therefore, it is always true that $\Pr(\neg(S^\ast\rightarrow v_1)\mid\{v_1\}\infect{S^\ast}w)\geq\frac1{2(k+1)}$.

Finally, we consider $\Sigma(\{v_1\})\setminus\Sigma(S^\ast)$ by only accounting for those vertices in $\Gamma(u_1)\setminus S^\ast$.
\begin{align*}
    \frac{|\Sigma(\{v_1\})\setminus\Sigma(S^\ast)|}{\prod_{w\in V}\deg(w)}&\geq\sum_{w\in\Gamma(u_1)\setminus S^\ast}\Pr\left(\left(\{v_1\}\infect{S^\ast}w\right)\land\neg\left(S^\ast\rightarrow v_1\right)\right)\\
    &\geq\sum_{w\in\Gamma(u_1)\setminus S^\ast}\frac1{2(k+1)}\Pr\left(\{v_1\}\infect{S^\ast}w\right)\\
    &>\frac1{2(k+1)}\cdot \frac{9\sigma(S^\ast)}{2000k^2}\tag{result from Step~1}\\
    &>\frac{1}{4000k^3}\frac{|\Sigma(S^\ast)|}{\prod_{w\in V}\deg(w)}.
\end{align*}
By Lemma~\ref{lem:maxc_outside}, this implies $|\Sigma(S)|\geq(1-(1-\frac1k)^k+\frac1{64000k^3})|\Sigma(S^\ast)|$, which further implies this proposition.
\end{proof}

Finally, we prove that $C_3$ cannot be too large if the greedy algorithm does not overcome the $1-(1-1/k)^k$ barrier.
Informally, this is because $C_3$ corresponds to a subset of the intersection among $\Sigma(\{u_1\}),\ldots,\Sigma(\{u_k\})$, and Lemma~\ref{lem:maxc_disjoint} implies that it cannot be too large.
\begin{proposition}\label{prop:C3}
If $C_3>\frac1{k^2}\cdot\sigma(S^\ast)$, then $\sigma(S)\geq(1-(1-\frac1k)^k+\frac1{8k^3})\sigma(S^\ast)$.
\end{proposition}
\begin{proof}
Notice that $C_3\prod_{w\in V}\deg(w)$ is at most the number of tuples $(w,g)$ such that $w$ is reachable from more than one vertex in $S^\ast$ under $g$.
It is easy to see that
$$C_3\prod_{w\in V}\deg(w)\leq \left(\sum_{i=1}^k|\Sigma(\{u_i\})|\right)-|\Sigma(S^\ast)|$$
because: 1) each $(w,g)$ such that $w$ is reachable by more than one vertex in $S^\ast$ under $g$ is counted at most once by $C_3\prod_{w\in V}\deg(w)$, exactly once by $\Sigma(S^\ast)$, and at least twice by $\sum_{i=1}^k\Sigma(\{u_i\})$, so the contribution of each such $(w,g)$ to the right-hand side of the inequality is at least the contribution of it to the left-hand side; 2) each $(w,g)$ such that $w$ is reachable by exactly one vertex in $S^\ast$ under $g$ is not counted by $C_3\prod_{w\in V}\deg(w)$ and is counted exactly once by both $\sum_{i=1}^k\Sigma(\{u_i\})$ and $\Sigma(S^\ast)$, so the contribution of such $(w,g)$ is the same on both sides of the inequality; 3) each $(w,g)$ such that $g$ is not reachable from $S^\ast$ contributes $0$ to both sides of the inequality.
Observing this inequality, if $C_3>\frac1{k^2}\cdot\sigma(S^\ast)$, we have
$$\left(\sum_{i=1}^k|\Sigma(\{u_i\})|\right)-|\Sigma(S^\ast)|>\frac1{k^2}\sigma(S^\ast)\prod_{w\in V}\deg(w)= \frac1{k^2}|\Sigma(S^\ast)|.$$
Lemma~\ref{lem:maxc_disjoint} implies $|\Sigma(S)|\geq (1-(1-\frac1k)^k+\frac1{8k^3})|\Sigma(S^\ast)|$, which implies this proposition.
\end{proof}

With Proposition \ref{prop:C1}, \ref{prop:C2} and \ref{prop:C3}, if $\sigma(S)=(1-(1-1/k)^k+o(1/k^3))\sigma(S^\ast)$, it must be that
$$\frac{|\Sigma(\{v_1\})\cap\Sigma(S^\ast)|}{\prod_{w\in V}\deg(w)}=C_1+C_2+C_3\leq\left(\frac1{k^2}+\frac9{10k}+\frac1{100k}\right)\sigma(S^\ast)<\frac{\frac{92}{100k}|\Sigma(S^\ast)|}{\prod_{w\in V}\deg(w)}.$$
However, Lemma~\ref{lem:maxc_firstintersection} would have implied $\sigma(S)\geq(1-(1-\frac1k)^k+\frac8{400k})\sigma(S^\ast)$, which is a contradiction.
This finishes proving Theorem~\ref{thm:lowerbound}.

\section{Alternative Models}
\label{sect:generalize}
In this section, we first consider the linear threshold \infmax on more general models.
Naturally, Theorem~\ref{thm:upperbound} holds if the model is more general.
We study if Theorem~\ref{thm:lowerbound} still holds.
We consider whether the barrier $1-(1-1/k)^k$ can still be overcome.  Subsequently, we consider alternative definitions of the linear threshold model.

\paragraph{Directed graphs}
If we consider \infmax with the linear threshold model on general directed graphs, Theorem~\ref{thm:lowerbound} no longer holds.
Moreover, for any positive function $f(k)$ which may be infinitesimal, there is always an example where the greedy algorithm achieves less than a $(1-(1-1/k)^k+f(k))$-approximation.
Example~\ref{example} can be easily adapted to show this.
Firstly, all the $k(k+\ell)$ edges $(u_i,v_j)$ become directed, so the cliques associated with those $u_i$'s are not even needed.
We replace each $C_i$ by a single vertex $u_i$.
Secondly, $D_1,\ldots,D_{k+\ell}$ become directed stars such that the directed edges in each star $D_i$ are from $v_i$ to the remaining vertices in the star.
Lastly, we change the size of the star so that $|D_i|=\lceil m(1-\frac1k)^{i-1}\rceil$ for $i=1,\ldots,k$ and $|D_{k+1}|=\cdots=|D_{k+\ell-1}|=\lceil m(1-\frac1k)^k\rceil$, where $\ell$ and $|D_{k+\ell}|$ are set such that $\sum_{i=1}^{k+\ell}|D_i|=mk-2k$ and $m$ is a large integer which can be set significantly larger than $1/f(k)$.

Now each $u_i$ has in-degree $0$, so will never be infected unless seeded.
Each $v_j$ has in-degree exactly $k$, and each $u_i$ will contribute $1/k$ to $v_j$'s infection probability.
Straightforward calculations reveal that the greedy algorithm will pick $S=\{v_1,\ldots,v_k\}$ so that $\sigma(S)=\sum_{i=1}^k\lceil m(1-\frac1k)^{i-1}\rceil\leq mk(1-(1-k)^k)+k$.
On the other hand, the optimal solution is $S^\ast=\{u_1,\ldots,u_k\}$, and $\sigma(S^\ast)=k+(mk-2k)=mk-k$.
We have $\frac{\sigma(S)}{\sigma(S^\ast)}=\frac{mk(1-(1-k)^k)+k}{mk-k}$, which can be less than $(1-(1-1/k)^k+f(k))$ when $m$ is sufficiently large.

\paragraph{Prescribed seed set}
\citet{khanna2014influence} considered the more generalized setting where the seed set $S$ can only be a subset of a prescribed vertex set $V'\subseteq V$, where $V'$ is a part of the input of the instance, and showed that their result for the independent cascade model can be extended to this setting.
It is straightforward to check that our proof for Theorem~\ref{thm:lowerbound} can also be extended to this setting.
In particular, all the lemmas in Sect.~\ref{sect:propertiesmaxc} hold for the generalized \maxc setting where $\calS$ must be a subset of a prescribed candidate set $\M'\subseteq\M$, with the proofs being exactly the same.
Basically, the proofs in Sect.~\ref{sect:lowerbound} do not rely on that each vertex in $V$ is a valid seed choice, so restricting that the seeds can only be chosen from $V'$ does not invalidate any propositions or lemmas.

\paragraph{Weighted vertices}
Another generalization \citet{khanna2014influence} considered is to allow that each vertex $v$ has a positive weight $\omega(v)$, and the objective of \infmax is to find the seed set that maximizes the expected total \emph{weighted} of infected vertices.
\citet{khanna2014influence} showed that the greedy algorithm can still achieve a $(1-(1-1/k)^k+c)$ approximation (for some constant $c>0$) for this generalized model.
We show that, for the linear threshold model, the story is completely different.
If vertices are weighted, for any positive function $f(k)$ which may be infinitesimal, there is always an example where the greedy algorithm achieves less than a $(1-(1-1/k)^k+f(k))$-approximation (for the linear threshold model \infmax with undirected graphs).
Thus, Theorem~\ref{thm:lowerbound} fails to extend to this setting.
While the settings with and without weighted vertices are not very different in the independent cascade model, they are quite different for the linear threshold model.

Again, Example~\ref{example} can be easily adapted to show our claim.
Let $m\gg k$ be a very large number.
Firstly, change the size of each clique $C_i$ to $m^{0.1}$.
Secondly, instead of connecting each $v_i$ to a lot of vertices to form a star, we let $v_i$ have a very high weight (so each star $D_i$ is replaced by a single vertex $v_i$).
Specifically, let $\omega(v_i)=m(1-1/k)^{i-1}$ for each $i=1,\ldots,k$, let $\omega(v_{k+1})=\cdots=\omega(v_{k+\ell-1})=m(1-1/k)^{k}$, and let $\ell$ and $\omega(v_{k+\ell})$ be such that $\sum_{i=1}^{k+\ell}\omega(v_i)=mk-m^{0.1}k$.
Let the weight of all the remaining vertices be $1$.
The greedy algorithm will pick $\{v_1,\ldots,v_k\}$, and the expected total weight of infected vertices is $o(m^{0.1})+\sum_{i=1}^{k}\omega(v_i)=m(1-(1-1/k)^k)+o(m^{0.1})$.
The optimal seeds are $u_1,\ldots,u_k$, with expected total weight of infected vertices being at least $mk-m^{0.1}k$.
We have $\frac{\sigma(S)}{\sigma(S^\ast)}\leq\frac{mk(1-(1-1/k)^k)+o(m^{0.1})}{mk-m^{0.1}k}$, which is less than $(1-(1-1/k)^k+f(k))$ when $m$ is sufficiently large.

\paragraph{Alternative models for linear threshold model}
We have defined the linear threshold model for \emph{unweighted}, undirected graphs where all the incoming edges of a vertex have the same weight.

In general, we can define the linear threshold model on edge-weighted directed graph $G=(V,E,w)$, where the weights satisfy the constraint that, for each vertex $v$, $\sum_{u\in\Gamma(v)}w(u,v)\leq1$.  A vertex $v$'s threshold $\theta_v$ is a real number sampled uniformly at random from the interval $[0,1]$, and $v$ is infected if the sum of the weights of the edges connecting from its infected neighbors exceeds the threshold: $\sum_{u\in\Gamma(v):u\text{ is infected}}w(u,v)\geq\theta_v$.

Notice that the constraint $\sum_{u\in\Gamma(v)}w(u,v)\leq1$ mentioned earlier is essential for the linear threshold model, as otherwise the probability a vertex $v$ is infected is no longer ``linear'' in terms of the influence from its infected in-neighbors, and the resultant model becomes fundamentally different.

Definition~\ref{def:LTM} is a special case of this model by assigning weights to the edges in the graph (that is originally unweighted) as follows: $w(u,v)=\frac1{\deg(v)}$.
Note that the assigned edge weights are not necessarily symmetric (i.e., $w(u, v) = w(v, u)$ is not necessarily true), as is common in past literature.  If the undirected graph  $G=(V,E,w')$ is weighted, then a natural extension is to define $w(u, v) = \frac{w'(u, v)}{\sum_{u \in \Gamma(v)} w'(u, v))}$.

In Appendix~\ref{sect:generalizeLTM}, we discuss alternative or more general ways to define a linear threshold model on undirected graphs.
In particular, we show that in the above undirected weighted version of the linear threshold model, the greedy algorithm cannot achieve a $(1-(1-1/k)^k+f(k))$-approximation for any positive function $f(k)$. (See the subsection ``weighted undirected graphs with normalization'' in Appendix~\ref{sect:generalizeLTM}.)
We also consider a version where we require all incoming edges of a vertex $v$ to have the same weight but the total weight is allowed to be strictly less than $1$.
In this case, all our results (Theorem~\ref{thm:upperbound} and Theorem~\ref{thm:lowerbound}) still hold.  (See the subsection ``Unweighted undirected graphs with slackness'' in Appendix~\ref{sect:generalizeLTM}.)

\section{Conclusion and Open Problems}
We have seen that the greedy algorithm for \infmax with the linear threshold model on undirected graphs can overcome the $1-(1-1/k)^k$ barrier by an additive term $\Omega(1/k^3)$ as shown in Theorem~\ref{thm:lowerbound}.
However, Theorem~\ref{thm:upperbound} suggests that, unlike the case for the independent cascade model, the greedy algorithm  cannot overcome the $(1-1/e)$ barrier for $k\rightarrow\infty$ for the linear threshold model.
Moreover, we have seen in Sect.~\ref{sect:generalize} that the approximation guarantee $1-(1-1/k)^k$ is tight if the vertices are weighted, which is different from Khanna and Lucier's result for the independent cascade model.
This again suggests that there are fundamental differences between these two diffusion models.

The tight instance in Example~\ref{example} has a significant limitation: it cannot scale to large $\sigma(S^\ast)$.
Notice that, to make the example work, we have to make the size of each $D_i$ be $o(k)$ and $\sigma(S^\ast)=o(k^2)$.
Otherwise, $\{u_1,\ldots,u_k\}$ will not be able to infect each $v_i$ with probability $1-o(1)$.
If the sizes of $D_i$ are $\omega(k)$, each seed in the seed set $\{v_1,\ldots,v_k\}$ output by the greedy algorithm will not be connected from $\{u_1,\ldots,u_k\}$ with a constant probability.
In the \maxc view, this will imply that $\Sigma(S)\setminus\Sigma(S^\ast)$ contains a significant number of elements, which will make the greedy algorithm overcome the $1-1/e$ barrier.
Therefore, a natural question is, if $\sigma(S^\ast)$ is large enough, say, $\sigma(S^\ast)=\omega(k^2)$, can the $1-1/e$ barrier be overcome?
We believe it can be overcome, and we make the following conjecture.
\begin{conjecture}
Consider \infmax problem $(G=(V,E),k)$ with the linear threshold model on undirected graphs.
If $\max_{S:S\subseteq V,|S|\leq k}\sigma(S)=\omega(k^2)$, there exists a constant $c>0$ such that the greedy algorithm achieves a $(1-1/e+c)$-approximation.
\end{conjecture}

Other than to prove (or disprove) the conjecture above, another open problem is to further close the gap for \infmax with undirected graphs.
Right now, the gap between $(1-1/e)$ \cite{KempeKT03} and $(1-\tau)$ \cite{schoenebeck2019influence} is still large.
Designing an approximation algorithm that achieves significantly better than a $(1-1/e)$-approximation and proving stronger APX-hardness results are two interesting and important directions for future work.

\section*{Acknowledgement}
We would like to thank the anonymous reviewers for their helpful and constructive comments.
Especially, we would like to thank the anonymous reviewer who brought us the point for the linear threshold model on edge-weighed undirected graphs, which motivates us to include the discussions in Append.~\ref{sect:generalizeLTM}.

\bibliographystyle{plainnat}
\bibliography{reference}

\newpage
\appendix
\section{Proof of Lemma~\ref{lem:clique}}
\label{append:proof_lem:clique}
Given a seed $s$ in the complete graph $G$, we calculate the probability that an arbitrary vertex $v\in V\setminus\{s\}$ is infected according to Remark~\ref{remark:RR}.
Consider the reverse random walk without repetition starting from $v$ as described in Remark~\ref{remark:RR}.
It reaches $s$ in one move with probability $\frac1{n-1}$, and it reaches $s$ in $t$ moves with probability
$\frac{1}{n-1}\prod_{i=1}^{t-1}\frac{n-1-i}{n-1}$ since $\prod_{i=1}^{t-1}\frac{n-1-i}{n-1}$ is the probability that the random walk never reaches $s$ and never comes back to any vertices that have been visited within the first $t-1$ moves and $\frac{1}{n-1}$ is the probability that the random walk moves to $s$ in the $t$-th move.
Putting this together, $v$ will be infected by $s$ with probability
$$\frac{1}{n-1}+\sum_{t=2}^{n-1}\frac{1}{n-1}\prod_{i=1}^{t-1}\frac{n-1-i}{n-1}=\frac1{n-1}\sum_{t=1}^{n-1}\prod_{i=0}^{t-1}\frac{n-1-i}{n-1}.$$
Simple calculations reveal an upper bound for this probability.
\begin{align*}
    &\frac1{n-1}\sum_{t=1}^{n-1}\prod_{i=0}^{t-1}\frac{n-1-i}{n-1}\\
    =&\frac1{n-1}\left(\sum_{t=1}^{\lceil\sqrt{n}\rceil}\prod_{i=0}^{t-1}\frac{n-1-i}{n-1}+\sum_{t=\lceil\sqrt{n}\rceil+1}^{n-1}\prod_{i=0}^{t-1}\frac{n-1-i}{n-1}\right)\\
    =&\frac1{n-1}\left(\sum_{t=1}^{\lceil\sqrt{n}\rceil}\prod_{i=0}^{t-1}\frac{n-1-i}{n-1}+\left(\prod_{i=0}^{\lceil\sqrt{n}\rceil-1}\frac{n-1-i}{n-1}\right)\sum_{t=\lceil\sqrt{n}\rceil+1}^{n-1}\prod_{i=\lceil\sqrt{n}\rceil}^{t-1}\frac{n-1-i}{n-1}\right)\tag{distributive law}\\
    <&\frac1{n-1}\left(\sum_{t=1}^{\lceil\sqrt n\rceil}1+1\cdot\sum_{t=1}^{n-1-\lceil\sqrt{n}\rceil}\left(\frac{n-1-\lceil\sqrt n\rceil}{n-1}\right)^t\right)\tag{the first two products are replaced by $1$, and $\prod_{i=\lceil\sqrt{n}\rceil}^{t-1}\frac{n-1-i}{n-1}\leq\left(\frac{n-1-\lceil\sqrt n\rceil}{n-1}\right)^{t-\lceil\sqrt{n}\rceil}$}\\
    <&\frac1{n-1}\left(\sum_{t=1}^{\lceil\sqrt n\rceil}1+\sum_{t=0}^{\infty}\left(\frac{n-1-\lceil\sqrt n\rceil}{n-1}\right)^t\right)\tag{the summation is extended to the infinite series}\\
    =&\frac1{n-1}\left(\left\lceil\sqrt{n}\right\rceil+\frac{n-1}{\lceil\sqrt n\rceil}\right).
\end{align*}
Finally, by linearity of expectation, the expected total number of infected vertices is
$$1+(n-1)\cdot \frac1{n-1}\left(\left\lceil\sqrt{n}\right\rceil+\frac{n-1}{\lceil\sqrt n\rceil}\right)<3\sqrt{n},$$
which concludes the lemma.

\section{Proofs in Sect.~\ref{sect:propertiesmaxc}}
We prove all the lemmas in Sect.~\ref{sect:propertiesmaxc} here.
Notice that the lemmas are restated for the ease of reading.
Again, for all the lemmas in this section, we are considering a \maxc instance $(U,\M,k)$ where $\calS=\{S_1,\ldots,S_k\}$ denotes $k$ subsets output by the greedy algorithm and $\calS^\ast=\{S_1^\ast,\ldots,S_k^\ast\}$ denotes the optimal solution.

We first define a useful notion called a \maxc instance \emph{with restriction}.

\begin{definition}
Given a \maxc instance $(U,\M,k)$ and a subset $U'\subseteq U$, the \maxc instance $( U,\M,k)$ \emph{with restriction on} $ U'$ is another \maxc instance $( U',\M',k')$ where $\M'=\{S\cap U':S\in\M\}$.
\end{definition}

We begin by proving the following lemma, which compares the $k$ subsets output by the greedy algorithm with arbitrary $\ell$ subsets.
This is a more general statement than saying that the greedy algorithm always achieves a $(1-(1-1/k)^k)$-approximation.

\begin{lemma}\label{lem:maxc_kl}
Given a \maxc instance $( U,\M,k)$, let  $\calS'$ be an arbitrary collection of $\ell$ subsets, we have $\val( \calS)\geq(1-(1-1/\ell)^k)\val( \calS')$.
\end{lemma}
\begin{proof}
Fix an arbitrary $\ell$, we prove this lemma by induction on $k$.
To prove the base step for $k=1$, the subset in $\calS'$ with the largest size covers at least $\frac1\ell\val( \calS')$ elements, so the first subset picked by the greedy algorithm should cover at least $\frac1\ell\val( \calS')$ elements.
Thus, for $k=1$, $\val( \calS)\geq(1/\ell)\val( \calS')=(1-(1-1/\ell)^k)\val( \calS')$.

For the inductive step, suppose this lemma holds for $k=k_0$, we aim to show that it holds for $k=k_0+1$.
Let $ \calS=\{S_1,\ldots,S_{k_0+1}\}$ be the output of the greedy algorithm.
By the same analysis above, $|S_1|\geq\frac1\ell\val( \calS')$.
Consider the \maxc instance $( U'= U\setminus S_1,\M',k_0)$ which is the instance $( U,\M,k_0+1)$ with restriction on $ U\setminus S_1$.
Since the greedy algorithm selects subsets based on \emph{marginal} increments to $\val(\cdot)$, $(S_2\setminus S_1),\ldots,(S_{k_0+1}\setminus S_1)$ will also be the $k_0$ subsets picked by the greedy algorithm on the restricted instance.
By the induction hypothesis, we have $$\left|\left(\bigcup_{i=2}^{k_0+1}S_i\right)\setminus S_1\right|=\left|\bigcup_{i=2}^{k_0+1}(S_i\setminus S_1)\right|\geq\left(1-\left(1-\frac1\ell\right)^{k_0}\right)\left|\bigcup_{S\in\calS'}(S\setminus S_1)\right|$$
$$\qquad=\left(1-\left(1-\frac1\ell\right)^{k_0}\right)\left|\left(\bigcup_{S\in\calS'}S\right)\setminus S_1\right|.$$
We then discuss two different cases.

If $|S_1\cap(\bigcup_{S\in \calS'}S)|\leq\frac1\ell|\bigcup_{S\in \calS'}S|=\frac1\ell\val( \calS')$, then $|(\bigcup_{S\in \calS'}S)\setminus S_1|\geq\frac{\ell-1}\ell\val( \calS')$ and
$$\val(\calS)=|S_1|+\left|\left(\bigcup_{i=2}^{k_0+1}S_i\right)\setminus S_1\right|\geq\frac1\ell\val( \calS')+\left(1-\left(1-\frac1\ell\right)^{k_0}\right)\cdot\frac{\ell-1}{\ell}\val( \calS')$$
$$\qquad=\left(1-\left(1-\frac1\ell\right)^{k_0+1}\right)\val( \calS'),$$
which concludes the inductive step.

If $|S_1\cap(\bigcup_{S\in \calS'}S)|>\frac1\ell|\bigcup_{S\in \calS'}S|=\frac1\ell\val( \calS')$, let $|S_1\cap(\bigcup_{S\in \calS'}S)|=(\frac1\ell+c)\val( \calS')$ for some $c\in(0,1-\frac1\ell]$, and we have
$$\val( \calS)=|S_1|+\left|\left(\bigcup_{i=2}^{k_0+1}S_i\right)\setminus S_1\right|\geq\left(\frac1\ell+c\right)\val( \calS')+\left(1-\left(1-\frac1\ell\right)^{k_0}\right)\cdot\left(1-\frac1\ell-c\right)\val( \calS')$$
$$\qquad=\left(1-\left(1-\frac1\ell\right)^{k_0+1}+c\left(1-\frac1\ell\right)^{k_0}\right)\val( \calS')>\left(1-\left(1-\frac1\ell\right)^{k_0+1}\right)\val( \calS'),$$
which concludes the inductive step as well.
\end{proof}

The lemma below shows that, if the first subset picked by the greedy algorithm is one of the subsets in the optimal solution, then the barrier $1-(1-1/k)^k$ can be overcome.

\maxcgo*
\begin{proof}
Assume $S_1=S_1^\ast$ without loss of generality.
In order to be picked by the greedy algorithm, $S_1^\ast$ should also be the subset in $\calS^\ast$ with the largest size.
Therefore, $|S_1|=|S_1\cap(\bigcup_{i=1}^kS_i^\ast)|=(\frac1k+c)\val(\calS^\ast)$ for some $c\geq0$, and $|(\bigcup_{i=2}^kS_i^\ast)\setminus S_1|=|(\bigcup_{i=1}^kS_i^\ast)\setminus S_1|=(1-\frac1k-c)\val(\calS^\ast)$.
By applying Lemma~\ref{lem:maxc_kl} on the instance with restriction $U\setminus S_1$, we have $|(\bigcup_{i=2}^kS_i)\setminus S_1|\geq(1-(1-\frac1{k-1})^{k-1})|(\bigcup_{i=2}^kS_i^\ast)\setminus S_1|=(1-(1-\frac1{k-1})^{k-1})(1-\frac1k-c)\val(\calS^\ast)$.
Putting together,
\begin{align*}
    \val(\calS)&=|S_1|+\left|\left(\bigcup_{i=2}^kS_i\right)\setminus S_1\right|\\
    &\geq\left(\frac1k+c\right)\val\left(\calS^\ast\right)+\left(1-\left(1-\frac1{k-1}\right)^{k-1}\right)\left(1-\frac1k-c\right)\val\left(\calS^\ast\right)\\
    &=\left(1-\left(1-\frac1k\right)\left(1-\frac1{k-1}\right)^{k-1}+c\left(1-\frac1{k-1}\right)^{k-1}\right)\val(\calS^\ast)\\
    &\geq\left(1-\left(1-\frac1k\right)\left(1-\frac1k\right)^{k-1}\left(1-\frac1{(k-1)^2}\right)^{k-1}\right)\val(\calS^\ast)\tag{since $\left(1-\frac1k\right)\left(1-\frac1{(k-1)^2}\right)=1-\frac1{k-1}$ and $c\left(1-\frac1{k-1}\right)^{k-1}\geq0$}\\
    &\geq\left(1-\left(1-\frac1k\right)^k\left(1-\frac1{(k-1)^2}\right)\right)\val(\calS^\ast)\\
    &\geq\left(1-\left(1-\frac1k\right)^k+\frac1{(k-1)^2}\left(1-\frac1k\right)^k\right)\val(\calS^\ast).
\end{align*}
The lemma follows from noticing $(1-\frac1k)^k\geq\frac14$ and $\frac1{k-1}>\frac1k$.
\end{proof}

Next, we show that, in order to have the tight approximation guarantee $1-(1-1/k)^k$, the first subset picked by the greedy algorithm must intersect almost exactly $1/k$ fraction of the elements covered by the $k$ optimal subsets.

\maxcfirstintersection*
\begin{proof}
By the same argument in the first paragraph of the proof of Lemma~\ref{lem:maxc_kl}, we have $|S_1|\geq\frac1k\val(\calS^\ast)$.
On the other hand, considering the instance with restriction on $U\setminus S_1$, the greedy algorithm, picking subsets based on marginal increments, will pick $(S_2\setminus S_1),\ldots,(S_k\setminus S_1)$ as the first $k-1$ seeds in the restricted instance.
Applying Lemma~\ref{lem:maxc_kl}, we have $|(\bigcup_{i=2}^kS_i)\setminus S_1|\geq(1-(1-1/k)^{k-1})|(\bigcup_{i=1}^kS_i^\ast)\setminus S_1|$.

If $\frac{|S_1\cap(\bigcup_{i=1}^kS_i^\ast)|}{\val(\calS^\ast)}>\frac1k+\varepsilon$, let $\frac{|S_1\cap(\bigcup_{i=1}^kS_i^\ast)|}{\val(\calS^\ast)}=\frac1k+c$ where $c>\varepsilon$.
The last paragraph of the proof of Lemma~\ref{lem:maxc_kl} can be applied here, and we have $$\val(\calS)\geq\left(\frac1k+c\right)\val(\calS^\ast)+\left(1-\left(1-\frac1k\right)^{k-1}\right)\left(1-\frac1k-c\right)\val(\calS^\ast)$$
$$\qquad=\left(1-\left(1-\frac1k\right)^k+c\left(1-\frac1k\right)^{k-1}\right)\val(\calS^\ast)>\left(1-\left(1-\frac1k\right)^k+\frac\varepsilon4\right)\val(\calS^\ast),$$
since $(1-\frac1k)^{k-1}>\frac14$ and $c>\varepsilon$.

If $\frac{|S_1\cap(\bigcup_{i=1}^kS_i^\ast)|}{\val(\calS^\ast)}<\frac1k-\varepsilon$, let $\frac{|S_1\cap(\bigcup_{i=1}^kS_i^\ast)|}{\val(\calS^\ast)}=\frac1k-c$ where $c\in(\varepsilon,\frac1k)$.
We have
$$\left|\left(\bigcup_{i=2}^kS_i\right)\setminus S_1\right|\geq\left(1-\left(1-\frac1k\right)^{k-1}\right)\left|\left(\bigcup_{i=1}^kS_i^\ast\right)\setminus S_1\right|$$
$$\qquad=\left(1-\left(1-\frac1k\right)^{k-1}\right)\left(1-\frac1k+c\right)\val(\calS^\ast).$$
Adding $S_1$, we have $$\val(\calS)=|S_1|+\left|\left(\bigcup_{i=2}^kS_i\right)\setminus S_1\right|\geq\frac1k\val(\calS^\ast)+\left(1-\left(1-\frac1k\right)^{k-1}\right)\left(1-\frac1k+c\right)\val(\calS^\ast)$$
$$\qquad=\left(1-\left(1-\frac1k\right)^k+c\left(1-\left(1-\frac1k\right)^{k-1}\right)\right)\val(\calS^\ast)>\left(1-\left(1-\frac1k\right)^k+\frac\varepsilon4\right)\val(\calS^\ast),$$ since $1-(1-\frac1k)^{k-1}>\frac14$ (this holds $k\geq2$; if $k=1$, the premise of the lemma will not hold as we will then have $S^\ast=\{S_1\}$)
and $c>\varepsilon.$
\end{proof}

The next lemma shows that, in order to have the tight approximation guarantee $1-(1-1/k)^k$, the first subset output by the greedy algorithm must not cover a number of elements that is significantly more than $1/k$ fraction of the number of elements in the optimal solution.

\begin{lemma}\label{lem:maxc_first}
If $|S_1|\geq(\frac1k+\varepsilon)\val(\calS^\ast)$ for some $\varepsilon>0$ which may depend on $k$, then $\val(\calS)\geq(1-(1-1/k)^k+\varepsilon/8)\val(\calS^\ast)$.
\end{lemma}
\begin{proof}
If $|S_1\cap(\bigcup_{i=1}^kS_i^\ast)|/\val(\calS^\ast)\notin[\frac1k-\frac\varepsilon2,\frac1k+\frac\varepsilon2]$, Lemma~\ref{lem:maxc_firstintersection} directly implies this lemma.
Suppose $|S_1\cap(\bigcup_{i=1}^kS_i^\ast)|/\val(\calS^\ast)\in[\frac1k-\frac\varepsilon2,\frac1k+\frac\varepsilon2]$.
Since $|S_1|\geq(\frac1k+\varepsilon)\val(\calS^\ast)$, we have $|S_1\setminus(\bigcup_{i=1}^kS_i^\ast)|>\frac\varepsilon2\val(\calS^\ast)$.
Let $|S_1\cap(\bigcup_{i=1}^kS_i^\ast)|=(\frac1k+c)\val(\calS^\ast)$ where $c\in[-\frac\varepsilon2,\frac\varepsilon2]$.
By the same analysis in the last paragraph of the proof of Lemma~\ref{lem:maxc_firstintersection} (which uses Lemma~\ref{lem:maxc_kl} as well),
$$\val(\calS)=|S_1|+\left|\left(\bigcup_{i=2}^kS_i\right)\setminus S_1\right|\geq\left(\frac1k+\varepsilon\right)\val(\calS^\ast)+\left(1-\left(1-\frac1k\right)^{k-1}\right)\left(1-\frac1k+c\right)\val(\calS^\ast)$$
$$=\left(1-\left(1-\frac1k\right)^k+c\left(1-\left(1-\frac1k\right)^{k-1}\right)+\varepsilon\right)\val(\calS^\ast)>\left(1-\left(1-\frac1k\right)^k+\frac\varepsilon8\right)\val(\calS^\ast).$$
For the last inequality, it holds trivially if $c\geq0$, and it holds for $c<0$ as $c(1-(1-\frac1k)^{k-1})+\varepsilon>c+\varepsilon\geq\frac\varepsilon2$.
\end{proof}

The next two lemmas show that, in order to have the tight approximation guarantee $1-(1-1/k)^k$,  those optimal subsets must be almost disjoint and the first subset output by the greedy algorithm must not cover too many elements that are not covered by the optimal subsets.

\maxcdisjoint*
\begin{proof}
If $\sum_{i=1}^k|S_i^\ast|>(1+\varepsilon)\val(\calS^\ast)$, the subset in $\calS^\ast$ with the largest size should contain more than $\frac{1+\varepsilon}{k}\val(\calS^\ast)$ elements, implying that $|S_1|>\frac{1+\varepsilon}{k}\val(\calS^\ast)$.
Lemma~\ref{lem:maxc_first} then implies that $\val(\calS)\geq(1-(1-\frac1k)^k+\frac\varepsilon{8k})\val(\calS^\ast)$.
\end{proof}

\maxcoutside*
\begin{proof}
If $|S_1\cap(\bigcup_{i=1}^kS_i^\ast)|<(\frac1k-\frac\varepsilon2)\val(\calS^\ast)$, Lemma~\ref{lem:maxc_firstintersection} implies this lemma. Otherwise, we have
$$|S_1|=\left|S_1\setminus\left(\bigcup_{i=1}^kS_i^\ast\right)\right|+\left|S_1\cap\left(\bigcup_{i=1}^kS_i^\ast\right)\right|\geq \varepsilon\val(\calS^\ast)+\left(\frac1k-\frac\varepsilon2\right)\val(\calS^\ast)=\left(\frac1k+\frac\varepsilon2\right)\val(\calS^\ast),$$
and Lemma~\ref{lem:maxc_first} implies this lemma.
\end{proof}

Finally, the lemma below shows that, in order to have the tight approximation guarantee $1-(1-1/k)^k$, those subsets in the optimal solution must have about the same size.

\maxcbalance*
\begin{proof}
Assume $|S_1^\ast|<(\frac1k-\varepsilon)\val(\calS^\ast)\leq(\frac1k-\varepsilon)\sum_{i=1}^k|S_i^\ast|$ without loss of generality.
We have
$$\sum_{i=2}^k|S_i^\ast|>\left(\frac{k-1}k+\varepsilon\right)\sum_{i=1}^k|S_i^\ast|\geq\left(\frac{k-1}k+\varepsilon\right)\val(\calS^\ast).$$
Therefore,
$$\max_{2\leq i\leq k}|S_i^\ast|>\left(\frac1k+\frac\varepsilon{k-1}\right)\val(\calS^\ast)>\left(\frac1k+\frac\varepsilon{k}\right)\val(\calS^\ast).$$
By the nature of the greedy algorithm, $$|S_1|\geq\max_{2\leq i\leq k}|S_i^\ast|>\left(\frac1k+\frac\varepsilon{k}\right)\val(\calS^\ast),$$
and Lemma~\ref{lem:maxc_first} implies this lemma.
\end{proof}

We remark that we only include those properties that are useful in our analysis, while there are some other important properties for \maxc that are not listed here.

\section{Alternative Models for Linear Threshold Model on Undirected Graphs}
\label{sect:generalizeLTM}
As mentioned in the last subsection of Sect.~\ref{sect:generalize}, we will discuss alternative or more general ways to define a linear threshold model on undirected graphs, and discuss whether our results in Sect.~\ref{sect:upperbound} and Sect.~\ref{sect:lowerbound} extend to those new settings.

\paragraph{Weighted undirected graphs with symmetric weights}
A seemingly natural way to define the linear threshold model on undirected edge-weighted graphs is to define edge-weighted undirected graphs such that the weights satisfy the constraints that, 1) for each vertex $v$, $\sum_{u\in\Gamma(v)}w(u,v)\leq1$ (as it is in the linear threshold model for general directed graphs), and 2) $w(u,v)=w(v,u)$ for any pair $\{u,v\}$ (so that the graph is undirected).
However, this model is unnatural in reality, because it disallows the case that a popular vertex exercises significant influence over many somewhat lonely vertices.
Consider an extreme example where the graph is a star, with a center $u$ and $n$ leaves $v_1,\ldots,v_n$.
The constraint $\sum_{i=1}^nw(v_i,u)\leq1$ implies that there exists at least one $v_i$ such that $w(v_i,u)\leq\frac1n$, and furthermore, $w(u,v_i)=w(v_i,u)\leq\frac1n$.
In this case, even if $u$ is the only neighbor of $v_i$, $u$ still has very limited influence to $v_i$ just because $u$ has a lot of other neighbors.
In reality, it is unnatural to assume that a node's being popular reduces its influence to its neighbors.

The linear threshold model constraint  $\sum_{u\in\Gamma(v)}w(u,v)\leq1$ mekes the above model with symmetrically weighted graphs unnatural.  Moreover, this constrains is particular to the linear threshold model.  For the independent cascade model which does not have this constraint, it is much more natural to consider graphs with symmetric edge weights $\forall\{u,v\}:w(u,v)=w(v,u)$, and this is indeed the model studied most often in the past literature, including Khanna and Lucier's work \cite{khanna2014influence}.

\paragraph{Weighted undirected graphs with normalization}
A more natural way to define the linear threshold model on graphs that are both edge-weighted and undirected is to start with an edge-weighted undirected graph $G=(V,E,w')$ without any constraint and then normalize the weight of each edge $(u,v)$ such that $w(u,v)=\frac{w'(u,v)}{\sum_{u'\in\Gamma(v)}w'(u',v)}$ and $w(v,u)=\frac{w'(u,v)}{\sum_{v'\in\Gamma(u)}w'(u,v')}$, as mentioned in the last subsection of Sect.~\ref{sect:generalize}.
After normalization, we have, for each $v\in V$, $\sum_{u\in\Gamma(v)}w(u,v)=1$, so this is a valid linear threshold model.
Notice that, after the normalization, the weights of the two anti-parallel directed edges $(u,v)$ and $(v,u)$ may be different.  Even though they had the same weight before the normalization (to maintain the undirected feature).
In the corresponding live-edge interpretation, each $v$ chooses one of its incoming edges to be ``live'' with probability proportional to the edge-weights (instead of choosing one uniformly at random as in Theorem~\ref{thm:LT_live}).

Theorem~\ref{thm:upperbound} holds naturally under this more generalized model.
However, Theorem~\ref{thm:lowerbound} no longer holds, and the barrier $1-(1-1/k)^k$ is tight even up to lower order terms: for any positive function $f(k)$ which may be infinitesimal, there is always an example where the greedy algorithm achieves less than a $(1-(1-1/k)^k+f(k))$-approximation.
Example~\ref{example} can be easily adapted to show this.
Let $m\gg k$ be a sufficiently large number that is divisible by $k^k$ such that both $m^{0.1}$ and $\sqrt{m}$ are integers.
Increase the size of $C_1,\ldots,C_k$ to $m^{0.1}$.
Increase the sizes of the stars such that $|D_i|= m(1-\frac1k)^{i-1}$ for $i=1,\ldots,k$ and $|D_{k+1}|=\cdots=|D_{k+\ell-1}|= m(1-\frac1k)^k$, where $\ell$ and $|D_{k+\ell}|$ are set such that $\sum_{i=1}^{k+\ell}|D_i|=km-k\sqrt{m}$.
Set the weights of the edges in each $D_i$ to be extremely small, say $1/m^{100}$, and set the weights of the remaining edges to be $1$.
After normalizing the weights, the weight of each edge connecting $v_i$ to each of the remaining vertices in $D_i$ is still $1$, the weight of each edge $(u_i,v_j)$ (for $i=1,\ldots,k$ and $j=1,\ldots,k+\ell$) becomes $\frac1{k+(|D_j|-1)/m^{100}}\approx\frac1k$, the weight of each edge $(v_j,u_i)$ (again for $i=1,\ldots,k$ and $j=1,\ldots,k+\ell$) becomes $\frac1{k+\ell+|C_i|-1}=\Theta(\frac1m)$ which can be made much smaller than $f(k)$.
By a similar argument, the greedy algorithm will choose $\{v_1,\ldots,v_k\}$, while the optimal seed set is $\{u_1,\ldots,u_k\}$.
We have $\frac{\sigma(S)}{\sigma(S^\ast)}\leq\frac{mk(1-(1-1/k)^k)+o(m^{0.1})}{km-k\sqrt{m}}$, which can be less than $(1-(1-1/k)^k+f(k))$ for sufficiently large $m$.

Lemma~\ref{lem:|A|1} and Lemma~\ref{lem:|A|2} rely crucially on the fact that each vertex $v$ should choose its incoming live edge \emph{uniformly at random}, and Lemma~\ref{lem:deg+1} also relies on this.
This explains why the proof of Theorem~\ref{thm:lowerbound} fails to work for this edge-weighted setting.

\paragraph{Unweighted undirected graphs with slackness}
In the previous setting, as well as the unweighted setting used in this paper, we have $\sum_{u\in\Gamma(v)}w(u,v)$ equals exactly $1$.
Equivalently, each $v$ chooses exactly one incoming live edge.
The most general linear threshold model allows that $\sum_{u\in\Gamma(v)}w(u,v)$ may be strictly less than $1$, or that each $v$ can choose no incoming live edge with certain probability.

To define a model that incorporates this feature, we consider a more general model where each vertex $v$ has a parameter $\vartheta_v\in[0,1]$ (given as an input to the algorithm) such that each vertex $v$ chooses no incoming live edge with probability $1-\vartheta_v$, and, with probability $\vartheta_v$, it chooses an incoming edge being live uniformly at random.
Equivalently, given an undirected unweighted graph $G=(V,E)$, we assign weights to the edges such that $w(u,v)=\frac{\vartheta_v}{\deg(v)}$ and consider the standard linear threshold model on directed graphs.
Notice that we could further generalize this to allow weighted graphs, and then normalize the weights of the edges such that the sum of the weights of all incoming edges of each vertex $v$ is exactly $\vartheta_v$.
However, this is a model that is even more general than the one in the last subsection (the model in the last subsection is obtained by setting $\vartheta_v=1$ for all $v$ from this model), and we know that the ratio $1-(1-1/k)^k$ is tight even up to infinitesimal additive $f(k)$.
Thus, in this subsection, we consider the unweighted setting with the $(1-\vartheta_v)$ slackness for each vertex $v$.
We will show that both Theorem~\ref{thm:upperbound} and Theorem~\ref{thm:lowerbound} hold under this setting.
It is clear that Theorem~\ref{thm:upperbound} holds, as we are considering a more general model.

To see that Theorem~\ref{thm:lowerbound} holds, we first observe that Lemma~\ref{lemma:negativecorrelation}, Lemma~\ref{lem:|A|1} and Lemma~\ref{lem:|A|2} hold with exactly the same proofs.
To see that the remaining part of the proof of Theorem~\ref{thm:lowerbound} can be adapted to this setting, we need to show that Lemma~\ref{lem:deg+1} holds, and we need to establish that \infmax under this setting is still a special case of \maxc so that Proposition~\ref{prop:C3}, \ref{prop:C1} and \ref{prop:C2} hold.

Note that Lemma~\ref{lem:deg+1} is also true for this new setting with slackness, and it can be proved by a simple coupling argument if knowing Lemma~\ref{lem:deg+1} for the original setting without slackness is true. Alternative, it can be proved directly by a similar arguments used by \citet{schoenebeck2019influence}, and we include such a proof in Appendix~\ref{append:extension_slackness} for completeness.

We will use a more general version of \maxc with weighted elements, where each element $e_i$ has a positive weight $w(e_i)$, and the objective function we are maximizing becomes $\displaystyle\val(\mathcal{S})=\sum_{e\in\bigcup_{S\in\mathcal{S}}S}w(e)$.
All the lemmas in Sect.~\ref{sect:propertiesmaxc} hold for the weighted \maxc with exactly the same proofs.
The interpretation of an \infmax instance to a \maxc instance is almost the same as it is given in Sect.~\ref{sect:premaxc}.
The elements are tuples in $V\times H$ where $H$ is the set of all possible realizations.
Notice that here $|H|=\prod_{v\in V}(\deg(v)+1)$, as an extra outcome that an vertex chooses no incoming live edge is possible now.
The weight of the element $(v,g)$ equals to the probability that $g$ is sampled.
Therefore, $\sigma(S)=\sum_{v\in V}\Pr(S\rightarrow v)=\sum_{v\in V}\sum_{g:\text{ }v\text{ is reachable from }S\text{ under }g}\Pr(g\text{ is sampled})=\sum_{(v,g):\text{ }v\text{ is reachable from }S\text{ under }g}w((v,g))$.
Let $\Sigma(S)$ be the same as before (which is the set of all elements $(v,g)$ that are ``covered'' by $S$, or equivalently, the set of all $(v,g)$'s such that $v$ is reachable from $S$ under $g$).
We have $\sigma(S)=\sum_{(u,g)\in\Sigma(S)}w((u,g))$.

Finally, Proposition~\ref{prop:C3}, \ref{prop:C1} and \ref{prop:C2} hold with the following changes to the proof.
\begin{itemize}
    \item every $|\Sigma(S)|$ is changed to its weighted version $\sum_{(u,g)\in\Sigma(S)}w((u,g))$;
    \item $\prod_{v\in V}\deg(v)$ is changed to $1$ (Notice that we had $|H|=\prod_{v\in V}\deg(v)$ before, but we have $\sum_{g\in H}\Pr(g\text{ is sampled})=1$ now).
\end{itemize}

\section{Proof of Lemma~\ref{lem:deg+1} Including Slackness}
\label{append:extension_slackness}
Recall that, in the linear threshold model on undirected graphs with slackness, each vertex has a parameter $\vartheta_v\in[0,1]$ that is given as an input to the algorithm.
With probability $1-\vartheta_v$, vertex $v$ chooses no incoming live edge, and with probability $\vartheta_v$, vertex $v$ chooses one of its incoming edges as the live edge uniformly at random. (See the last subsection of Append.~\ref{sect:generalizeLTM}.)

We prove the following lemma in this section.

\begin{lemma}\label{lem:deg+1slack}
Consider the linear threshold model on undirected graphs with slackness.
For any $v\in V$, we have $\sigma(\{v\})=\deg(v)+1$.
\end{lemma}

This lemma is a generalization to Lemma~\ref{lem:deg+1}, as the linear threshold model in Definition~\ref{def:LTM} used in this paper is a special case with the slackness of each vertex being $0$.

This lemma also fills in the last piece of the proof that Theorem~\ref{thm:lowerbound} holds for the setting with unweighted undirected graphs with slackness.

As mentioned, we  the arguments is largely identical to the one by \citet{schoenebeck2019influence}.

We first show that Lemma~\ref{lem:deg+1} holds for trees.

\begin{lemma}\label{lem:LT_1seed}
Suppose $G$ is a tree, we have $\sigma(\{v\})\leq\deg(v)+1$.
\end{lemma}
\begin{proof}
We assume without loss of generality that $G$ is rooted at $v$.
Consider an arbitrary vertex $u\neq v$ at the penultimate level with children $u_1,\ldots,u_t$ being leaves of $T$.
We have $\deg(u)=t+1$.
Suppose $u$'s parent $s$ is infected by $v$ with probability $x$ ($x=1$ if $s=v$).
Then $u$ will be infected with probability $\frac{x\vartheta_u}{t+1}$, and each $u_i$ of $u_1,\ldots,u_t$, having degree $1$, will be infected with probability $\vartheta_{u_i}$ if $u$ is infected.
Therefore, the expected number of infected vertices in the subtree rooted at $u$ is $$\frac{x\vartheta_u}{t+1}\left(1+\sum_{i=1}^t\vartheta_{u_i}\right)+\left(1-\frac{x\vartheta_u}{t+1}\right)\cdot0x\leq\frac{x\vartheta_u}{t+1}(t+1)=x\vartheta_u.$$
This suggests that, if we contract the subtree rooted at $u$ to a single vertex $u$, the expected total number of infected vertices can only increase for this change of the graph $G$, since the degree of $u$ becomes $1$ after this contraction, making the infection probability of $u$ become $x\vartheta_u$.
We can keep doing this contraction until $G$ becomes a star with center $v$, and the expected number of infected vertices can only increase during this process.
The lemma follows.
\end{proof}

We define the \emph{lift} of an undirected graph $G$ with respect to a vertex $a\subseteq V$, which is a new undirected graph $\widehat{G}_a$ that shares the same vertex $a$ with $G$ plus a lot of new vertices.
We will then define a coupling between sampling live-edges in $G$ and sampling live-edges in $\widehat{G}_a$.
Given the seed $v$, this coupling reveals an upper bound of $\sigma(\{v\})$.
In particular, we will show $\sigma_G(\{v\})\leq\sigma_{\widehat{G}_v}(\{v\})$, where $\sigma_G(\cdot)$ and $\sigma_{\widehat{G}_v}(\cdot)$ denote the function $\sigma(\cdot)$ with respect to the graphs $G$ and $\widehat{G}_v$ respectively.

Let
$$\calP_a=\left\{P=((v_1,v_2),(v_2,v_3),\ldots,(v_{t-1},v_t)): v_1=a; v_2,\ldots,v_t\neq a;\forall i\neq j:v_i\neq v_j\right\}$$
be the set of all simple paths $P$ that start from vertex $a$ but never come back to $a$.

\begin{definition}\label{def:lift}
Given an undirected graph $G=(V,E)$ and $a\in V$, the \emph{lift} of $G$ with respect to $a$, denoted by $\widehat{G}_a=(\widehat{V},\widehat{E})$, is an undirected graph defined as follows.
\begin{itemize}
    \item The vertex set is $\widehat{V}=\{a\}\cup V_P$, where $V_P=\{v_P:P\in\calP_a\}$ is the set of vertices corresponding to the simple paths in $\calP_a$.
    \item For each $v_P\in V_P$, include $(a,v_P)\in\widehat{E}$ if $P$ is a path of length $1$ that starts from $a$; for each $v_{P_1},v_{P_2}\in V_P$, include $(v_{P_1},v_{P_2})$ if $|P_2|=|P_1|+1$ and $P_2,P_1$ share the first $|P_1|$ common edges (or $|P_1|=|P_2|+1$ and $P_1,P_2$ share the first $|P_2|$ common edges, since $\widehat{G}_a$ is undirected).
    \item If $P\in\calP_a$ is a path ending at a vertex in $G$ that is adjacent to $a$, add a dummy vertex in $\widehat{G}$ and connect this vertex to $v_P$.
\end{itemize}
\end{definition}

It is easy to see that $\widehat{G}_a$ is a tree (that can be viewed as) rooted at $a$.
The vertices in the tree $\widehat{G}_a$ correspond to all the paths in $\calP_a$ starting at $a$.
For any path $P\in\calP_a$ with $v$ being its ending vertex, $\deg(v_P)$ in $\widehat{G}_a$ equals to $\deg(v)$ in $G$.

Let $\Psi:E\to2^{\widehat{E}}$ be the function mapping an undirected edge in $G$ to its counterparts in $\widehat{G}_a$:
$$\Psi(e)=\left\{
\begin{array}{ll}
    \{(a,v_P)\mid P=((a,v))\} & \mbox{if }e=(a,v)\\
    \{(v_{P_1},v_{P_2})\mid P_2=(P_1,e)\} & \mbox{Otherwise.}
\end{array}
\right.$$
Notice that in the above definition, $\Psi(e)$ contains only a single edge $(a,v_P)$ with $P=((a,v))$ being the length-one path connecting $a$ and $v$ if $e=(a,v)$, while $\Psi(e)$ contains the set of all $(v_{P_1},v_{P_2})$ such that $P_2$ is obtained by appending $e$ to $P_1$.
Let $\Phi:V\to 2^{\widehat{V}}$ represent the vertex correspondence:
$$\Phi(v)=\left\{
\begin{array}{ll}
    \{v\} &  \mbox{if }v=a \\
    \{v_P\mid P\mbox{ ends at }v\} & \mbox{Otherwise.}
\end{array}
\right.$$

From our definition, it is easy to see that $\Psi(e_1)\cap\Psi(e_2)=\emptyset$ if $e_1\neq e_2$, and $\Phi(u)\cap\Phi(v)=\emptyset$ if $u\neq v$.
Moreover, since $\calP_a$ contains only paths, for any vertex $v$ and edge $e$ in $G$, each path in $\widehat{G}_{a}$ connecting $a$ to a leaf (recall that $\widehat{G}_{a}$ is a tree) can intersect each of $\Psi(e)$ and $\Phi(v)$ at most once.\footnote{To see this for each $\Psi(e)$, suppose for the sake of contradiction that the path from $v_P$ to the root contains two edges $(v_{P_1},v_{P_2}), (v_{P_3},v_{P_4})$ such that $(v_{P_1},v_{P_2}), (v_{P_3},v_{P_4})\in\Psi(e)$ for some edge $e$.
Assume without loss of generality that the order of the four vertices on the path according to the distances to the root is $(v_{P_1},v_{P_2},v_{P_3},v_{P_4})$.
It is easy to see from our construction that $P_1\subsetneq P_2\subsetneq P_3\subsetneq P_4$.
As a result, $(v_{P_1},v_{P_2}), (v_{P_3},v_{P_4})\in\Psi(e)$ implies that $P_2$ is the path obtained by appending $e$ to $P_1$, and $P_4$, containing $P_2,P_3$, is obtained by appending $e$ to $P_3$, which further implies that $P_4$ is a path that uses the edge $e$ twice, contradicting to our definition that $\calP_a$ contains only simple paths.

The corresponding claim for each $\Phi(v)$ can be shown similarly.}

Finally, to let the inequality $\sigma_G(\{v\})\leq\sigma_{\widehat{G}_v}(\{v\})$ make sense, we need to specify the parameter $\vartheta$ for each vertex in $\widehat{G}_v$.
This is done in a natural way: for each vertex $w\in\Phi(v)$ in $\widehat{G}_v$, set $\vartheta_w$ for vertex $w$ in $\widehat{G}_v$ be the same as $\vartheta_v$ for vertex $v$ in $G$.

\begin{lemma}\label{lem:LT_coupling}
$$\sigma_G(\{v\})\leq\sigma_{\widehat{G}_v}(\{v\}).$$
\end{lemma}
\begin{proof}
We will define a coupling between the process of revealing live-edges in $G$ and the process of revealing live-edges in $\widehat{G}_v$.
Let $\chi_G$ be the edge-revelation process in $G$, and $\chi_{\widehat{G}_v}$ be the edge-revelation process in $\chi_{\widehat{G}_v}$, where in both processes, each edge is viewed as two anti-parallel directed edges, and we always reveal all the incoming edges for a vertex $u$ simultaneously by choosing exactly one incoming edge uniformly at random with probability $\vartheta_u$.
We will couple $\chi_G$ with another edge-revelation process $\chi_{\widehat{G}_v}'$ of $\widehat{G}_v$.

We consider the following coupling.
In each iteration where all the incoming edges of $u$, denoted by $(u_1,u),(u_2,u),\ldots,(u_{\deg(u)},u)$, are revealed such that at most one of them is live, we reveal all the incoming edges for each $v_P\in\Phi(u)$ as follows.
\begin{itemize}
    \item If none of $(u_1,u),(u_2,u),\ldots,(u_{\deg(u)},u)$ is live in $G$, then $v_P$ chooses no live incoming edge.
    \item For each $P'$ such that $v_{P'}$ is a neighbor of $v_P$, there must exists $u_i\in\{u_1,\ldots,u_{\deg(u)}\}$ such that either that $P'$ is obtained by appending $(u,u_i)$ to $P$ or that $P$ is obtained by appending $(u_i,u)$ to $P'$. Reveal the directed edge $(v_{P'},v_P)$ such that it is live if and only if $(u_i,u)$ is live in $G$.
    \item If there is a live edge $(v_{P'},v_P)$ revealed in the above step, make all the remaining directed edges connecting to $v_P$ not be live. If no live edge is revealed in the above step and one of $(u_1,u),(u_2,u),\ldots,(u_{\deg(u)},u)$ in $G$ is live, it must be that $(a,u)$ is an edge in $G$ and $u$ has chosen $(a,u)$ being the live edge. In this case, let the edge between $v_P$ and the dummy vertex being live (See the third bullet point of Definition~\ref{def:lift}).
\end{itemize}
This defines a coupling between $\chi_G$ and $\chi_{\widehat{G}_v}'$.
It is easy to check that each $v_P\in\widehat{V}$ chooses exactly one of its incoming edges uniformly at random with probability $\vartheta_{v_P}$ and chooses no incoming edge with probability $1-\vartheta_{v_P}$ in this coupling, which is the same as it is in the process $\chi_{\widehat{G}_v}$.
The difference is that, there are dependencies between the revelations of incoming edges for different vertices in $\widehat{G}_v$: if both $v_P,v_{P'}\in\widehat{V}$ belongs to the same $\Phi(u)$ for some $u\in V$, the incoming edges for $v_P$ and $v_{P'}$ are revealed in the same way.

Although the two processes $\chi_{\widehat{G}_v}'$ and $\chi_{\widehat{G}_v}$ are not the same, we will show that the expected number of vertices that are reachable from $v$ by live edges is the same in both $\chi_{\widehat{G}_v}'$ and $\chi_{\widehat{G}_v}$.
It suffices to show that, for each $v_P\in\widehat{V}$, all the \emph{vertices} in the path connecting $v_P$ to the seed $v$ are considered independently (meaning that the incoming edges for $v_{P_1}$ on the path are revealed independently to the revelations of the incoming edges of $v_{P_2}$), since this would imply that the probability $v_P$ is connected to a seed is the same in both $\chi_{\widehat{G}_v}'$ and $\chi_{\widehat{G}_v}$, and the total number of vertices reachable from $v$ by live edges is the same by the linearity of expectation.
We only need to show that there do not exist two vertices on this path that are in the same set $\Phi(u)$ for some $u\in V$, since the incoming edges of each $v_{P_1}\in\Phi(u_1)$ are revealed independently to the revelations of the incoming edges of each $v_{P_2}\in\Phi(u_2)$ whenever $u_1\neq u_2$.
This is true due to that all the paths in $\calP_v$ are simple paths, as remarked in the paragraph below where we define function $\Phi(\cdot)$.

Following the same analysis before, we can show that the number of the vertices reachable from $v$ in $\chi_G$ is always at most the number of vertices reachable from $v$ in $\chi_{\widehat{G}_v}'$.
The lemma concludes here.
\end{proof}

Since $v$ has the same degree in $G$ and $\widehat{G}_v$, Lemma~\ref{lem:LT_coupling} and Lemma~\ref{lem:LT_1seed} implies Lemma~\ref{lem:deg+1slack}.

\end{document}